\documentclass[10pt, final, journal, letterpaper, twocolumn]{IEEEtran}
\makeatletter
\def\ps@headings{%
\def\@oddhead{\mbox{}\scriptsize\rightmark \hfil \thepage}%
\def\@evenhead{\scriptsize\thepage \hfil \leftmark\mbox{}}%
\def\@oddfoot{}%
\def\@evenfoot{}}
\makeatother 

\IEEEoverridecommandlockouts
\usepackage{bbm}
\usepackage{amsfonts}
\usepackage[dvips]{graphicx}
\usepackage{times}
\usepackage{cite}
\usepackage{amsmath}
\usepackage{array}
\usepackage{amssymb}

\usepackage{stfloats}
\usepackage{slashbox}
\usepackage{graphicx}
\usepackage{footnote}
\usepackage{booktabs}
\usepackage{array}
\usepackage{algorithm}
\usepackage{subeqnarray}
\usepackage{cases}
\usepackage{threeparttable}
\usepackage{color}
\usepackage{hyperref}
\usepackage{epstopdf}
\usepackage{algpseudocode}
\usepackage{bm}
\usepackage{multirow}
\usepackage[labelformat=simple]{subcaption}
\usepackage{adjustbox}

\newtheorem{theorem}{Theorem}

\newtheorem{proposition}{Proposition}
\newtheorem{lemma}{Lemma}

\begin{document}

\title{Ergodic Achievable Rate Maximization of RIS-assisted Millimeter-Wave MIMO-OFDM Communication Systems}

\author{\authorblockN{Renwang Li, Shu Sun, \IEEEmembership{Member,~IEEE}, Meixia Tao, \IEEEmembership{Fellow,~IEEE}}\\
\thanks{This work is supported  by the National Key R\&D Project of China under grant 2020YFB1406802. (Corresponding authors: Shu Sun, Meixia Tao.)}	
\thanks{The authors are with Department of Electronic Engineering, Shanghai Jiao Tong University, Shanghai, China (emails:\{renwanglee, shusun, mxtao\}@sjtu.edu.cn).}
}

\maketitle

\begin{abstract}
Reconfigurable intelligent surface (RIS) has attracted extensive attention in recent years. However, most research focuses on the scenario of the narrowband and/or instantaneous channel state information (CSI), while wide bandwidth with the use of millimeter-wave (mmWave) (including sub-Terahertz) spectrum is a major trend in next-generation wireless communications, and statistical CSI is more practical to obtain in realistic systems. Thus, we {consider} the ergodic achievable rate of RIS-assisted mmWave multiple-input multiple-output  orthogonal frequency division multiplexing communication systems. The widely used Saleh-Valenzuela channel model is adopted to  characterize the mmWave channels and only the statistical CSI is available. We first derive the approximations of the ergodic achievable rate by means of the majorization theory and Jensen's inequality. Then, an alternating optimization based algorithm is proposed to  maximize the ergodic achievable rate by jointly designing the transmit covariance matrix at the base station and the reflection coefficients at the RIS. Specifically, the design of the transmit covariance matrix is transformed into a power allocation problem and solved by spatial-frequency water-filling. The reflection coefficients are optimized by the Riemannian conjugate gradient algorithm. {Simulation results corroborate  the effectiveness of the proposed algorithms.}
\end{abstract}


\section{Introduction}
The millimeter wave (mmWave) communication over the 30-300 GHz spectrum is one of the most promising techniques for 5G-and-beyond systems \cite{6515173, 7400949}. However, the free-space path loss is more severe at mmWave compared to the conventional microwave bands. Typically, multiple-input multiple-output (MIMO) beamforming technology is employed to provide high gains to extend the transmission distance, but the high directivity makes the mmWave communication more sensitive to signal blockage. Meanwhile, MIMO technology greatly increases the consumption of power and cost. Recently, one promising and cost-effective solution to overcome these issues is to deploy reconfigurable intelligent surfaces (RISs) \cite{8910627,9122596,9229054}. An RIS is an artificial uniform planar array (UPA) with plenty of elements, each of which can independently impose a phase shift on the incident signal and then reflect or refract it passively with the assistance of a smart controller. Hence, by adaptively adjusting the coefficients, RIS can be controlled to enhance the transmission quality of the desired signals. RIS is spectrum- and energy- efficient since it does not require radio frequency components. In addition, RIS can be flexibly and widely deployed so as to enhance the coverage of the mmWave communication.

Motivated by the above promising advantages, RIS has attracted extensive attention in both academia and industry, e.g., \cite{8811733,8982186,9110912,9234098, huawei}. The authors in \cite{8811733} consider a power minimization problem under multiple-input single-ouput (MISO) scenario and propose a semidefinite relaxation based algorithm to jointly optimize the active and passive beamforming, while the weighted sum-rate maximization problem is studied in \cite{8982186}. The authors in \cite{9110912} focus on the capacity maximization problem under MIMO scenario and propose an alternating optimization (AO) based algorithm. In \cite{9234098}, the inherent sparse feature of the mmWave channels is exploited to find an efficient algorithm to jointly design the transceiver and RIS. It is mentioned in \cite{huawei} that the beamforming at the BS and the reflection coefficients at the RIS can be simultaneously optimized to improve the system performance.

However, all of the above works mainly focus on the narrowband communication systems, while 5G and future 6G communications are likely to conduct wideband deployment. Two main effects, i.e., spatial-wideband effect and frequency-wideband effect, emerge in mmWave MIMO wideband systems, which will dramatically affect the system performance \cite{8354789,ning2021prospective}. {The spatial-wideband effect refers to the phenomenon of a non-negligible time delay across the array aperture for the same symbol in wideband systems. To combat the spatial-wideband effect, orthogonal frequency division multiplexing (OFDM)  is a promising technology that divides the baseband into several sub-bands so that each sub-band can be considered frequency-independent \cite{8354789,ning2021prospective}.} As for the frequency-wideband effect which is also known as beam squint effect, there are no effective solutions thus far to the authors' best knowledge, especially for RIS \footnote{{The beam squint effect has been addressed in some literature, e.g., \cite{9409636, 9417413}. In order to mitigate the influence of the beam squint effect, the authors in \cite{9409636} propose a twin-stage orthogonal matching pursuit algorithm during channel estimation, while the authors in \cite{9417413} propose an efficient beamforming algorithm to maximize the achievable rate. Nevertheless, a general method to address the beam squint effect is still absent.}}. Specifically, the array response vectors differ across frequencies, notably for large bandwidth systems. However, RIS can only impose the same phase shifts on different frequencies since it is applied in the time domain and lacks the ability of digital signal processing \cite{9039554}. Some work has been devoted to investigating the wideband RIS-assisted communication systems \cite{9039554,8964457,8937491,9610122, 9417413, 9120639, nuti2021spectral, 8964330, hong2022hybrid, 9459505, 9685734, 9520295 }.
For the single-input single-output (SISO) scenario, the authors in \cite{9039554} aim at the achievable rate maximization and propose a successive convex approximation  based method for jointly power allocation and reflection coefficient optimization, while the problem of maximizing the minimum rate of all users is considered in \cite{8964457}. {The authors in \cite{8937491} first execute channel estimation, and then maximize the average achievable rate based on the strongest signal path. The authors in \cite{9610122} propose a low-complexity majorization-minimization-based algorithm to efficiently maximize the achievable rate. } For the MISO scenario, the authors in \cite{9417413} focus on mitigating the beam squint effect and propose low-complexity solutions for both line-of-sight (LoS) and non-LoS scenarios. Multi-user MISO scenario is studied in  \cite{9120639}, where the original sum rate maximization problem is reformulated as a modified mean square error  minimization problem, followed by a block coordinate descent iterative algorithm. For the MIMO scenario, the authors in \cite{nuti2021spectral} aim at maximizing the spectral efficiency for point-to-point communication, while hybrid digital and analog beamforming is considered in \cite{8964330,hong2022hybrid}. The weighted sum-rate maximization problem is studied in \cite{9459505}, where single data stream transmission is considered and the fractional programming is adopted to decouple the original problem. In addition, the secrecy rate maximization problem is investigated in \cite{9685734} and an AO-based inexact block coordinate descent algorithm by leveraging Lagrange multiplier and complex circle manifold methods is proposed, while the discrete reflecting phase shift case is extended in \cite{9520295}.

However, all of the above contributions are based on the instantaneous channel state information (CSI), which is pretty challenging to acquire in practice since RIS is usually nearly passive and composed of a large number of unit cells \cite{8910627,9122596,9229054}. An attractive alternative is to explore the statistical CSI. Therefore, we concentrate on a point-to-point RIS-assisted mmWave MIMO-OFDM communication system by exploiting the statistical CSI  where the RIS is adopted as a reflecting surface.  The main contributions of this study are summarized as follows:
\begin{itemize}
  \item  We focus on maximizing the ergodic achievable rate of the RIS-assisted mmWave MIMO-OFDM communication system. The widely used Saleh-Valenzuela (SV) channel model is adopted to characterize the mmWave channels and only the statistical CSI is assumed to be available. To the best of the authors' knowledge, this is the first effort to consider the statistical CSI in the broadband RIS-assisted mmWave MIMO-OFDM communication systems.
  \item We derive closed-form approximations of the ergodic achievable rate by means of the majorization theory and Jensen's inequality. The results show that the ergodic achievable rate increases logarithmically with the number of antennas at the base station (BS) and the user, the number of  reflection units at the RIS, the power allocation at the BS, and the eigenvalues of the steering matrices associated with the BS, RIS and user.
  \item The ergodic achievable rate is maximized by jointly designing the transmit covariance matrix at the BS and the reflection coefficients at the RIS. Specifically, the design of the transmit covariance matrix is transformed into a power allocation problem and solved by spatial-frequency water-filling. The reflection coefficients are optimized by the Riemannian conjugate gradient (RCG) algorithm.
\end{itemize}

Furthermore, we conduct extensive simulations to validate the performance of the proposed algorithms. {It is shown that our proposed alternating AO outperforms the benchmark scheme from the literature \cite{9234098} .} In addition, the ergodic achievable rate after optimization can be improved by about 25 bps/Hz, which emphasizes the importance of the optimization of the transmit covariance matrix at the BS and the reflection coefficients at the RIS.
Additionally, this work is a substantial improvement from {our previous work \cite{li2022risassisted}}. On the one hand, this work takes the direct BS-user link into consideration, while it is assumed to be blocked due to unfavorable propagation conditions in \cite{li2022risassisted}. When the direct link exists, the majorization theory cannot be applied directly because the cross product term of the ergodic achievable rate is no longer a positive semidefinite Hermitian matrix. Fortunately, the asymptotic orthogonality of the array response vectors can be utilized to overcome this hurdle. On the other hand, this work extends the narrowband scenario in \cite{li2022risassisted} into the broadband systems. The spatial- and frequency- wideband effects are mitigated in our proposed algorithm.

The rest of the paper is organized as follows. Section II presents the system model and the problem formulation. The approximations of the ergodic achievable rate are derived in Section III. An alternating optimization algorithm for maximizing the ergodic achievable rate is proposed in Section IV. The comparison with the state-of-the-art algorithms is discussed in Section V. Simulation results are provided in Section VI. Conclusions are drawn in Section VII.

\emph{Notations}: The imaginary unit is denoted by $j=\sqrt{-1}$.  Vectors and matrices are denoted by bold-face lower-case and upper-case letters, respectively. The conjugate, transpose and conjugate transpose of the vector $\bf x$ are denoted by $ \bf x^*$, $\mathbf{x}^T$ and $\mathbf x^H$, respectively. The $\odot$ symbol is the element-wise product and $\mathbb{E}(\cdot)$ is the expectation operation. The $\operatorname{tr}(\cdot)$, $\operatorname{det}(\cdot)$ and $\operatorname{rank}(\cdot)$ denote the trace, determinant and rank operation, respectively. The $\Re \{\cdot\}$ operation extracts the real value of a complex variable. The $\operatorname{unt}(\mathbf{x})$ operation represents an $N$-dimensional vector with elements $\frac{x_1}{|x_1|}, \cdots, \frac{x_N}{|x_N|}$. The distribution of a circularly symmetric complex Gaussian  random vector with mean vector $\mu$ and covariance matrix $\Sigma$ is denoted by $\mathcal{C}\mathcal{N}(\mu,\Sigma)$; and $\sim$ stands for ``distributed as''. The exponential random variable ${X}$ with parameter $\lambda$ is given by ${X} \sim \exp (\lambda)$.

\section{System Model and Problem Formulation}
In this section, we will first introduce the system model, then describe the widely adopted mmWave channel model, and finally propose the considered problem of maximizing the ergodic achievable rate.
\subsection{System Model}
\begin{figure}[t]
\begin{centering}
\includegraphics[width=.45\textwidth]{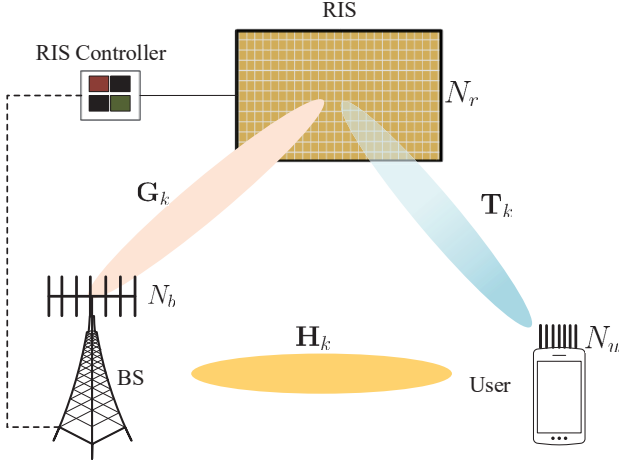}
 \caption{Illustration of an RIS-aided mmWave MIMO-OFDM system.} \label{system_model}
\end{centering}
\vspace{-0.5cm}
\end{figure}

We consider a downlink wideband point-to-point mmWave MIMO communication system as illustrated in Fig.~\ref{system_model}, where the BS is equipped with $N_b$ antennas and transmits signals to a user with $N_u$ antennas, via the help of one $N_r$-element RIS. In order to mitigate the spatial-wideband effect of the wideband mmWave channel, OFDM with $K$ subcarriers is employed to modulate the signal at the BS \cite{8354789, ning2021prospective}. Let $\mathbf{s}_k\in \mathbb{C}^ {N_b \times 1}$ denote the transmitted signal vector at subcarrier $k$. The transmit covariance matrix at subcarrier $k$ is defined as $\mathbf{Q}_k \triangleq \mathbb{E} \{\mathbf{s}_k \mathbf{s}_k^H \}\in \mathbb{C}^{N_b \times N_b} $, with $\mathbf{Q}_k \succeq \mathbf{0}$. The normalized transmit power constraint at the BS is given by
\begin{equation}\label{budgetP}
\sum \limits_{k=1}^K \operatorname{tr} \mathbb{E}\{\mathbf{s}_k \mathbf{s}_k^H\} = \sum \limits_{k=1}^K \operatorname{tr}(\mathbf{Q}_k ) \leq 1.
\end{equation}
Without loss of generality, we assume the maximum excess delay is of $L_\text{max}$ taps in the time domain for the baseband equivalent channels of both the BS-RIS-user reflecting link and the BS-user direct link. The transmit signal $\mathbf{s}_k$ is first transformed into the time domain through a $K$-point inverse discrete Fourier transform (IDFT), followed by the appending of a cyclic prefix (CP) of length $N_\text{CP}$, with $N_\text{CP} \geq L_\text{max}$. At the user side, after removing the CP and performing the $K$-point discrete Fourier transform (DFT), the baseband received signal $\mathbf{y}_k\in \mathbb{C}^{N_u \times 1}$ at the $k$-th subcarrier can be presented as
\begin{equation}
\mathbf{y}_k=(\mathbf{T}_k\mathbf{\Theta}\mathbf{G}_k+ \mathbf{H}_k) \mathbf{s}_k + \mathbf{n}_k,
\end{equation}
where $\mathbf{T}_k\in \mathbb{C}^{N_u \times N_r}$, $\mathbf{G}_k\in \mathbb{C}^{N_r \times N_b}$, and $\mathbf{H}_k\in \mathbb{C}^{N_u \times N_b}$ denote the frequency-domain channel on subcarrier $k$ from the RIS to the user, from the BS to the RIS, and from the BS to the user, respectively; $\mathbf{\Theta} \in \mathbb{C}^{N_r\times N_r}$ denotes the response matrix at the RIS, which is written as
\begin{equation}
\mathbf{\Theta} \triangleq \operatorname{diag} \{\xi_1 e^{j\theta_1}, \xi_2 e^{j\theta_2}, \ldots, \xi_{N_r} e^{j\theta_{N_r}}\},
\end{equation}
where $\theta_i\in [0, 2\pi)$ and $\xi_i \in[0, 1]$ represent the phase shift and the amplitude of the reflection coefficient of the $i$-th reflection unit, respectively; and $\mathbf{n}_k \thicksim\mathcal{C}\mathcal{N}{(\mathbf{0}, \sigma^2 \mathbf{I}_{N_u})}$ represents the additive white Gaussian noise vector with zero mean and variance $\sigma^2$ \footnote{Without loss of generality, the variances of noise at different subcarriers are assumed to be the same.}. Note that the RIS can only bring the same phase shift on different subcarriers due to the lack of the ability of digital signal processing  \footnote{{In this paper, we assume that the RIS units are frequency independent as widely adopted in the existing literature \cite{9409636,9039554,8964457,8937491,9610122, 9417413, 9120639, nuti2021spectral, 8964330, hong2022hybrid, 9459505, 9685734, 9520295 }. Some literature, e.g., \cite{9389801}, studies the property of the frequency dependency of RIS. Nevertheless, there is still lack of a general model to accurately characterize the nonlinear response of RIS over various frequencies. } } \cite{9039554}. {In addition, the amplitude $\xi_i, \forall i \in \mathcal{N}_r$ is assumed to be one to maximize the signal reflection and simplify the hardware design of the RIS \cite{8811733,8937491,9409636}.} Let $\mathcal{K} = \{1,\ldots,K\}$ and $\mathcal{N}_r = \{ 1,\ldots,N_r\}$ denote the index sets of subcarriers and RIS elements, respectively.

\subsection{Channel Model} \label{channel}
Due to the limited scattering paths in mmWave channels, the SV model is widely adopted for modelling the wideband mmWave channel. Suppose uniform linear arrays (ULAs) are equipped at the BS and the user, and a uniform planar array (UPA) is equipped at the RIS. The frequency domain channel matrix $\mathbf{H}_k$, $\mathbf{G}_k$ and $\mathbf{T}_k$ at subcarrier $k$ can be expressed as \cite{8354789, ning2021prospective,7397861,8844787,8794743}
\begin{equation} \label{channelH}
\mathbf{H}_k= \sqrt{\frac{N_b N_{u}}{L_h}} \sum_{i=1}^{L_h} \alpha_{h,i} e^{-j2\pi \tau_{h,i} f_k} \mathbf{a}_r(\psi_{h,u,i},f_k) \mathbf{a}_t^H (\psi_{h,b,i},f_k),
\end{equation}
\begin{equation} \label{channelG}
\mathbf{G}_k= \sqrt{\frac{N_r N_{b}}{L_g}} \sum_{i=1}^{L_g} \alpha_{g,i} e^{-j2\pi \tau_{g,i} f_k} \mathbf{a}_r(\phi_{g,r,i}, \varphi_{g,r,i}, f_k) \mathbf{a}_t^H (\psi_{g,b,i},f_k),
\end{equation}
\begin{equation} \label{channelT}
\mathbf{T}_k= \sqrt{\frac{N_r N_{u}}{L_t}} \sum_{i=1}^{L_t} \alpha_{t,i} e^{-j2\pi \tau_{t,i} f_k} \mathbf{a}_r(\psi_{t,u,i},f_k) \mathbf{a}_t^H (\phi_{t,r,i},\varphi_{t,r,i},f_k),
\end{equation}
where $L_h$ $(L_g, L_t)$ is the number of the paths between the BS and the user (the BS and the RIS, the RIS and the user); $\alpha_{h,i} \sim \mathcal{CN} (0,\sigma^2_{h,i})$ $(\alpha_{g,i} \sim \mathcal{CN} (0,\sigma^2_{g,i}), \alpha_{t,i} \sim \mathcal{CN} (0,\sigma^2_{t,i})$ \footnote{In mmWave MIMO communication systems, the complex gain is widely assumed to be i.i.d. random variables following the complex Gaussian distribution \cite{7397861,6847111}.} denotes the complex gain of the $i$-th path and assume $\sigma^2_{h,1} \geq \sigma^2_{h,2} \geq \ldots \geq \alpha_{h,L_h}^2$ $(\sigma^2_{g,1} \geq \sigma^2_{g,2} \geq \ldots \geq \alpha_{g,L_g}^2$, $\sigma^2_{t,1} \geq \sigma^2_{t,2} \geq \ldots \geq \alpha_{t,L_t}^2)$; $\tau_{h,i}$ $(\tau_{g,i}, \tau_{t,i})$ is the delay of the $i$-th path; {$f_k$ is the frequency at subcarrier $k$, which is given by}
\begin{equation}
f_k = f_c+ \frac{f_s}{K}\left(k-1-\frac{K-1}{2}\right), \forall k \in \mathcal{K},
\end{equation}
where $f_c$ is the carrier frequency and $f_s$ is the bandwidth. Note that in the narrowband communication systems, it is assumed that $f_k =f_c$. However, this assumption is invalid in wideband communication systems and will dramatically affect the ergodic achievable rate. $\mathbf{a}_r(\cdot)$  $(\mathbf{a}_t(\cdot))$ denotes the normalized array response vector at the receiver side (the transmitter side); $\psi_{h,u,i}$ ($\psi_{t,u,i}$) denotes the angle of arrival (AoA) associated with the user of the direct link (the reflecting link); $\psi_{h,b,i}$ ($\psi_{g,b,i}$) denotes the angle of departure (AoD) associated with the BS of the direct link (the reflecting link); {$\phi_{g,r,i}$ ($\phi_{t,r,i}$) and $\varphi_{g,r,i}$ ($\varphi_{t,r,i}$) denote the azimuth and elevation angles of the arrival (departure) associated with the RIS.} For a ULA with $N$ antennas at the frequency $f_k$, the normalized array response vector is given by
\begin{equation}
\mathbf{a}(\psi, f_k) = \frac{1}{\sqrt{N}}\left[ 1, e^{j\frac{2\pi f_k d}{c} \sin \psi},  \cdots,  e^{j \frac{2\pi f_k d}{c} (N-1) \sin \psi} \right]^T,
\end{equation}
where $c$ represents the speed of light, $d=\frac{c}{2f_c} = \frac{\lambda}{2}$ is the antenna spacing, $\lambda$ represents the signal wavelength of the central frequency, and $\psi$ denotes the AoA or AoD. For a UPA with $M=M_y \times M_z$ elements at the frequency $f_k$, the normalized array response vector can be expressed as
 \begin{equation}
 	\begin{aligned} 	
        \mathbf{a} \left(\phi, \varphi, f_k\right)= &
        \frac{1}{\sqrt{M}} \left[1, \cdots, e^{j \frac{2 \pi f_k d_r}{c} \left(m_y \sin \phi \sin \varphi + m_z \cos \varphi\right)}\right. , \\
        &
        \left.\cdots, e^{j \frac{2 \pi f_k d_r}{c} \left((M_y-1) \sin \phi \sin \varphi)+(M_z-1) \cos \varphi\right)}\right]^{T},
    \end{aligned}
  \end{equation}
{where $d_r$ is the unit cell spacing which is assumed to be half wavelength,} $\phi$ and $\varphi$ denote the azimuth and elevation angles, respectively.
Defining
$ \mathbf{A}_{uh,k} = \left[ \mathbf{a}_r(\psi_{h,u,1},f_k), \ldots, \mathbf{a}_r(\psi_{h,u,L_h}, f_k) \right] \in\mathbb{C}^{N_u \times L_h}$,
$\mathbf{A}_{bh,k} = \left[ \mathbf{a}_t(\psi_{h,b,1}, f_k),  \mathbf{a}_t(\psi_{h,b,2}, f_k), \ldots, \mathbf{a}_t(\psi_{h,b,L_h}, f_k) \right] \in\mathbb{C}^{N_b \times L_h}$,
and $\mathbf{H}_{L,k} =\sqrt{\frac{N_u N_{b}}{L_h}} \text{diag}( \alpha_{h,1} e^{-j 2 \pi \tau_{h,1}f_k}, \linebreak[4] \alpha_{h,2} e^{-j 2\pi \tau_{h,2}f_k}, \ldots, \alpha_{h,L_h} e^{-j 2\pi\tau_{h,L_h} f_k})$, the channel matrix $\mathbf{H}_k$ in \eqref{channelH} can be rewritten as
\begin{equation} \label{channelHk}
  \mathbf{H}_k = \mathbf{A}_{uh,k} \mathbf{H}_{L,k} \mathbf{A}_{bh,k}^H.
\end{equation}
Defining
$ \mathbf{A}_{rg,k} = [ \mathbf{a}_r(\phi_{g,r,1}, \varphi_{g,r,1}, f_k),\linebreak[4] \mathbf{a}_r(\phi_{g,r,2}, \varphi_{g,r,2},f_k),  \ldots, \mathbf{a}_r(\phi_{g,r,L_g}, \varphi_{g,r,L_g}, f_k) ] \in\mathbb{C}^{N_r \times L_g}$,
$\mathbf{A}_{bg,k} = \left[ \mathbf{a}_t(\psi_{g,b,1},f_k), \ldots, \mathbf{a}_t(\psi_{g,b,L_g},f_k) \right] \in\mathbb{C}^{N_b \times L_g}$, and
$\mathbf{G}_{L,k} = \sqrt{\frac{N_r N_{b}}{L_g}} \text{diag}( \alpha_{g,1} e^{-j2\pi \tau_{g,1}f_k}, \alpha_{g,2} e^{-j2\pi \tau_{g,2}f_k}, \linebreak[4]\ldots, \alpha_{g,L_g} e^{-j2\pi \tau_{g,L_g} f_k})$, the channel matrix $\mathbf{G}_k$ in \eqref{channelG} can be expressed as
\begin{equation} \label{channelGk}
  \mathbf{G}_k = \mathbf{A}_{rg,k} \mathbf{G}_{L,k} \mathbf{A}_{bg,k}^H.
\end{equation}
Defining $\mathbf{A}_{rt,k} = [ \mathbf{a}_t(\phi_{t,r,1}, \varphi_{t,r,1},f_k),\linebreak[4]  \mathbf{a}_t(\phi_{t,r,2}, \varphi_{t,r,2}, f_k), \ldots, \mathbf{a}_t(\phi_{t,r,L_t}, \varphi_{t,r,L_t},f_k) ] \in\mathbb{C}^{N_r \times L_t}$,
$ \mathbf{A}_{ut,k} = [ \mathbf{a}_r(\psi_{t,u,1},f_k), \ldots, \mathbf{a}_r(\psi_{t,u,L_t},f_k) ] \in\mathbb{C}^{N_u \times L_t}$, and
$ \mathbf{T}_{L,k} = \sqrt{\frac{N_r N_{u}}{L_t}} \text{diag}( \alpha_{t,1} e^{-j2\pi \tau_{t,1} f_k}, \alpha_{t,2} e^{-j2\pi \tau_{t,2} f_k},\linebreak[4] \ldots, \alpha_{t,L_t} e^{-j 2\pi\tau_{t,L_t}f_k})$,
the channel matrix $\mathbf{T}_k$ in \eqref{channelT} can be presented as
\begin{equation} \label{channelTk}
\mathbf{T}_k=  \mathbf{A}_{ut,k} \mathbf{T}_{L,k} \mathbf{A}_{rt,k}^H.
\end{equation}

\subsection{Problem Formulation}
For a given transmit covariance matrix $\{\mathbf{Q}_k\}_{k=1}^K$ and a given RIS response matrix $\mathbf{\Theta}$, the ergodic achievable rate of the RIS-assisted mmWave MIMO-OFDM communication system is given by
\begin{equation} \label{capacity_ori}
	\begin{aligned}
	&R\left(\{\mathbf{Q}_k\}_{k=1}^K, \mathbf{\Theta}\right) \\
	&= \mathbb{E}_{\mathbf{H}_\text{eff}} \left[ \frac{1}{K+N_\text{cp}}\sum \limits_{k=1}^K \log _{2} \operatorname{det}\left(\mathbf{I}_{N_{u}}+\frac{P_T}{\sigma^{2}} \mathbf{H}_{\text{eff},k} \mathbf{Q}_k \mathbf{H}_{\text{eff},k}^{H}\right)\right],
	\end{aligned}
\end{equation}
where $P_T>0$ is the power budget at the BS, and $\mathbf{H}_{\text{eff},k} = \mathbf{T}_k \mathbf{\Theta} \mathbf{G}_k + \mathbf{H}_k$ denotes the effective channel at subcarrier $k$ between the BS and the user. The expectation is taken over the effective channel $\mathbf{H}_{\text{eff}}$. More specifically, it is taken over the complex channel gains $\{\alpha_{h,i}\}_{i=1,2,\ldots,L_h}$ in Eq. \eqref{channelH}, $\{\alpha_{g,i}\}_{i=1,2,\ldots,L_g}$ in Eq. \eqref{channelG}, and $\{\alpha_{t,i}\}_{i=1,2,\ldots,L_t}$ in Eq. \eqref{channelT}.  In this paper, we aim to maximize the ergodic achievable rate by jointly designing the transmit covariance matrix $\{\mathbf{Q}_k\}_{k=1}^K$ at the BS and the response matrix $\mathbf{\Theta}$ at the RIS, subject to the maximum power budget at the BS. Therefore, the optimization problem can be formulated as
    \begin{subequations}\label{prob_original}
    \begin{align}
            \mathcal{P}_0: \max \limits_{ \{\mathbf{Q}_k\}_{k=1}^K, \mathbf{\Theta} } \quad & R\left(\{\mathbf{Q}_k\}_{k=1}^K, \mathbf{\Theta}\right) \\
            \text { s.t. } \quad\quad & \sum \limits_{k=1}^K \operatorname{tr}(\mathbf{Q}_k) \leq 1, \\
            & \mathbf{Q}_k \succeq 0, \forall k \in \mathcal{K}, \\
            &  \mathbf{\Theta} =\operatorname{diag}\left(e^{j \theta_{1}}, e^{j \theta_{2}}, \cdots, e^{j \theta_{N_r}} \right). \label{const_ris}
    \end{align}
    \end{subequations}
The problem $\mathcal{P}_0$ is highly challenging mainly due to the following three facts. First of all, there is no explicit expression of the ergodic achievable rate and the expectation operation in Eq. \eqref{capacity_ori} prevents further optimization. Secondly, the problem is non-convex due to the unit-modulus constraints of the RIS, and thus difficult to be optimally solved. Last but not the least, different from the narrowband communication systems, the RIS can only apply the same phase shifts over different subcarriers. Thereby, we need to consider all of the subcarriers when designing the RIS. In the following, we will firstly find an explicit expression to approximate the ergodic achievable rate in Section \ref{sec_app}, and then maximize the ergodic achievable rate by jointly designing $\{\mathbf{Q}_k\}_{k=1}^K$ and $\mathbf{\Theta}$ in Section \ref{sec_opt}.

\section{Approximations of the Ergodic Achievable Rate} \label{sec_app}
In this section, the approximation  of the ergodic achievable rate is firstly derived by means of the majorization theory. The Jensen's inequality is then used to obtain explicit and compact expressions. Finally, we evaluate the tightness of the approximations via numerical results.

\subsection{Approximation of the Ergodic Achievable Rate}
The derivation of the ergodic achievable rate is much more difficult than that in {our previous work \cite{li2022risassisted}}, where the direct link is blocked by obstacles. In the previous work,
\cite[Theorem 9.H.1.a]{marshall1979inequalities} is frequently adopted during the derivation. However, when the direct BS-user link exists, the cross product term $\mathbf{H}_k \mathbf{Q}_k \mathbf{G}_k^H \mathbf{\Theta}^H \mathbf{T}_k^H$ of the ergodic achievable rate is not a positive semi-definite Hermitian matrix, which prevents further derivation. Fortunately, we have the following proposition.
\begin{proposition} \label{prop_orthogo}
When $N_b$ goes to infinity, we have
\begin{equation}
  \mathbf{G}_k \mathbf{H}_k^H \rightarrow \mathbf{0}^{N_r \times N_u}.
\end{equation}
\end{proposition}

\begin{proof}
See Appendix \ref{append_pro1}.
\end{proof}

\textbf{Proposition} \ref{prop_orthogo} implies that the cross product term of the ergodic achievable rate can be neglected, when the transmit covariance matrix $\mathbf{Q}_k$ is an identity matrix and the number of the antennas at the BS is sufficiently large. Under  \textbf{Proposition} \ref{prop_orthogo} and the majorization theory\cite{marshall1979inequalities}, we have the following lemma.

\begin{lemma} \label{theorem_app}
Under the wideband SV channel model expressed in \eqref{channelHk}, \eqref{channelGk}, and \eqref{channelTk}, when the number of antennas at the BS goes to infinity, the ergodic achievable rate of the RIS-assisted mmWave MIMO-OFDM communication systems can be approximated by
\begin{equation}
{R}\left(\{\mathbf{Q}_k\}_{k=1}^K, \mathbf{\Theta}\right) \approx \widetilde{R}\left(\{\mathbf{Q}_k\}_{k=1}^K, \mathbf{\Theta}\right),
\end{equation}
\begin{figure*}[htb]
	\begin{equation} \label{ergo_app}
		\begin{aligned}
			\widetilde{R}\left(\{\mathbf{Q}_k\}_{k=1}^K, \mathbf{\Theta}\right)
			\triangleq \; & \mathbb{E}_{\alpha_{g}, \alpha_{t}} \left[ \frac{1}{K+N_{cp}} \sum \limits_{k=1}^K \sum \limits_{i=1}^{N_{s1}} \log_2 \left( 1+ \frac{P_T}{\sigma^2} \frac{N_b N_u N_r^2}{L_g L_t} q_{k,i} d_{ut,k,i} d_{bg,k,i} d_{r,k,i} |\alpha_{g,i}|^2 |\alpha_{t,i}|^2 \right) \right] \\
			& + \mathbb{E}_{\alpha_{h}} \left[ \frac{1}{K+N_{cp}} \sum \limits_{k=1}^K \sum \limits_{i=1}^{N_{s2}} \log_2 \left( 1+ \frac{P_T}{\sigma^2} \frac{N_b N_u}{L_h} q_{k,N_{s1}+i} d_{uh,k,i} d_{bh,k,i} |\alpha_{h,i}|^2 \right) \right],
		\end{aligned}
	\end{equation}
\end{figure*}
where $\widetilde{R}\left(\{\mathbf{Q}_k\}_{k=1}^K, \mathbf{\Theta}\right)$ is defined in Eq. \eqref{ergo_app}, shown at the top of the page, $N_{s1} = \min \Big( \operatorname{rank}\left(\mathbf{A}_{bg,k}^H \mathbf{A}_{bg,k}\right), \linebreak[4] \operatorname{rank}\left(\mathbf{A}_{ut,k}^H \mathbf{A}_{ut,k}\right), \operatorname{rank}\left(\mathbf{X}_{r,k}^H \mathbf{X}_{r,k}\right) \Big)$ denotes the rank of the BS-RIS-user reflecting link, $\mathbf{X}_{r,k} = \mathbf{A}_{rt,k}^H \mathbf{\Theta} \mathbf{A}_{rg,k}$; $ \small{ N_{s2} = \min \left( \operatorname{rank} \left(\mathbf{A}_{bh,k}^H \mathbf{A}_{bh,k}\right), \operatorname{rank} \left(\mathbf{A}_{uh,k}^H \mathbf{A}_{uh,k}\right) \right) }$ denotes the rank of the BS-user direct link, $N_s=N_{s1}+N_{s2}$;   $( d_{ut,k,1}, d_{ut,k,2}, \ldots, d_{ut,k,N_{s1}} )$, $( d_{bg,k,1}, d_{bg,k,2}, \ldots, d_{bg,k,N_{s1}} )$, $( d_{r,k,1}, d_{r,k,2}, \ldots, d_{r,k,N_{s1}} )$, $( d_{uh,k,1}, d_{uh,k,2}, \ldots, d_{uh,k,N_{s1}} )$, and  $( d_{bh,k,1}, d_{bh,k,2}, \ldots, d_{bh,k,N_{s1}} )$ are descending ordered eigenvalues of $\mathbf{A}_{ut,k}^H \mathbf{A}_{ut,k}$, $\mathbf{A}_{bg,k}^H \mathbf{A}_{bg,k}$, $\mathbf{X}_{r,k}^H \mathbf{X}_{r,k}$, $\mathbf{A}_{uh,k}^H \mathbf{A}_{uh,k}$ and $\mathbf{A}_{bh,k}^H \mathbf{A}_{bh,k}$, respectively; $q_{k,i}$ is the power allocated to the $i$-th data stream, whose value depends on the transmit covariance matrix $\mathbf{Q}_k$. We therefore explain the design of $\mathbf{Q}_k$ next.  Physically, when there are $L$ different paths, the BS can transmit at most $L$ data streams. Then, we can firstly align the transmit covariance matrix towards these paths and allocate different power over them later. Therefore, one reasonable and sub-optimal design of  $\mathbf{Q}_k$ is to determine its directions according to $\mathbf{A}_{bg,k}$ and $\mathbf{A}_{bh,k}$. Specifically, define the singular value decomposition (SVD) of $\mathbf{A}_{bg,k}$ and $\mathbf{A}_{bh,k}$ as
\begin{subequations}
\begin{align}
& \mathbf{A}_{bg,k} = \mathbf{U}_2 \boldsymbol{\Sigma}_2 \mathbf{V}_2^H,\\
& \mathbf{A}_{bh,k} = \mathbf{U}_3 \boldsymbol{\Sigma}_3 \mathbf{V}_3^3.
\end{align}
\end{subequations}
Let $\mathbf{U}_4 = \mathbf{U}_2(:,1:N_{s1}) \in \mathbb{C}^{N_b \times N_{s1}}$ denotes the first $N_{s1}$ columns of $\mathbf{U}_2$, $\mathbf{U}_5 = \mathbf{U}_3(:,1:N_{s2}) \in \mathbb{C}^{N_b \times N_{s2}}$ denotes the first $N_{s2}$ columns of $\mathbf{U}_3$, and $ \mathbf{U}_{q,k} = [\mathbf{U}_4, \mathbf{U}_5] \in \mathbb{C}^{N_b \times N_s}$. Then, the transmit covariance matrix $\mathbf{Q}_k$ is given by
\begin{equation} \label{approx_qk}
  \mathbf{Q}_k = \mathbf{U}_{q,k} \boldsymbol{\Lambda}_{q,k} \mathbf{U}_{q,k}^H,
\end{equation}
where $\boldsymbol{\Lambda}_{q,k} = \operatorname{diag}(q_{k,1}, q_{k,2}, \ldots, q_{k,N_s})$ and $q_{k,i}$ represents the amount of power allocated to the $i$-th data stream.
\end{lemma}
\begin{proof}
See Appendix \ref{append_theorem1}.
\end{proof}

Note that the approximation of the ergodic achievable rate  consists of two parts, i.e., the BS-RIS-user reflection part and the BS-user direct part. The reflection part depends on the channel gains $\{\alpha_g\}$ and $\{\alpha_t\}$ of the first $N_{s1}$ paths in both the BS-RIS link and the RIS-user link, while the direct part depends on the channel gains $\{\alpha_h\}$ of the first $N_{s2}$ paths.

\subsection{Jensen's Approximation of the Ergodic Achievable Rate}
The expectation operation in \eqref{ergo_app} hinders further optimization. To overcome this difficulty, Jensen's inequality is adopted to obtain simplified expressions.
\begin{theorem}
Under the wideband SV channel model expressed in \eqref{channelHk}, \eqref{channelGk}, and \eqref{channelTk}, when the number of antennas at the BS goes to infinity, the ergodic achievable rate of the RIS-assisted mmWave MIMO-OFDM communication systems can be upper bounded by
\begin{equation}
 \widetilde{R}\left(\{\mathbf{Q}_k\}_{k=1}^K, \mathbf{\Theta}\right) \leq R_{Jen}\left(\{\mathbf{Q}_k\}_{k=1}^K, \mathbf{\Theta}\right),
\end{equation}
where $R_{Jen}\left(\{\mathbf{Q}_k\}_{k=1}^K, \mathbf{\Theta}\right)$ is defined in Eq. \eqref{jen}, shown at the top of the next page.
\begin{figure*}[htb]
\begin{equation} \label{jen}
\begin{aligned}
R_{Jen}\left(\{\mathbf{Q}_k\}_{k=1}^K, \mathbf{\Theta}\right) &\triangleq  \frac{1}{K+N_{cp}} \sum \limits_{k=1}^K \sum \limits_{i=1}^{N_{s1}} \log_2 \left( 1+ \frac{P_T}{\sigma^2} \frac{N_b N_u N_r^2 \sigma_{g,i}^2 \sigma_{t,i}^2}{L_g L_t} q_{k,i} d_{ut,k,i} d_{bg,k,i} d_{r,k,i} \right) \\
 & \quad + \frac{1}{K+N_{cp}} \sum \limits_{k=1}^K \sum \limits_{i=1}^{N_{s2}} \log_2 \left( 1+ \frac{P_T}{\sigma^2} \frac{N_b N_u \sigma_{h,i}^2}{L_h} q_{k,N_{s1}+i} d_{uh,k,i} d_{bh,k,i} \right).
\end{aligned}
\end{equation}
\end{figure*}
\end{theorem}

\begin{proof}
See Appendix \ref{append_theorem3}.
\end{proof}

The Jensen's approximation  has a very compact form composed of the product of the eigenvalues and the summation over data streams. {Note that the above derivation only relies on the asymptotic orthogonality and eigenvalues of the steering matrices, thus it holds for both ULA and UPA structures.}  The Jensen's approximation is composed of the BS-RIS-user reflection part and the BS-user direct part. The reflection part is determined by signal-to-noise ratio (SNR), the number of antennas of the BS and the user, the number of the reflection units of the RIS, the power allocation at the BS, and the eigenvalues of the matrices of $\mathbf{A}_{ut,k}^H \mathbf{A}_{ut,k}$, $\mathbf{A}_{bh,k}^H \mathbf{A}_{bh,k}$ and $\mathbf{X}_{r,k}^H \mathbf{X}_{r,k}$. The direct part is influenced by SNR, the number of antennas of the BS and the user, the power allocation at the BS, and the eigenvalues of the matrices of $\mathbf{A}_{uh,k}^H \mathbf{A}_{uh,k}$ and $\mathbf{A}_{bh,k}^H \mathbf{A}_{bh,k}$.

In addition, when no CSI is available at the BS, uniform power allocation over transmit signals, i.e., $\mathbf{Q}_k = \frac{1}{KN_b} \mathbf{I}_{N_b}$, is optimal \cite{1611096}. In this case, we have the following lemma.
\begin{lemma}
Under the wideband SV channel model expressed in \eqref{channelHk}, \eqref{channelGk} and \eqref{channelTk}, when the number of antennas at the BS goes to infinity and $\mathbf{Q}_k =  \frac{1}{KN_b} \mathbf{I}_{N_b}$, the Jensen's approximation can be expressed as
\begin{equation} \label{Qid}
	\begin{aligned}
		&R_{Jen} \left(\mathbf{\Theta} \bigg| \mathbf{Q}_k= \frac{1}{KN_b} \mathbf{I}_{N_b} \right) \\
		& = \frac{1}{K+N_{cp}} \sum \limits_{k=1}^K \left[ \sum \limits_{i=1}^{N_{s2}} \log_2 \left( 1+ \frac{P_T}{\sigma^2} \frac{N_u \sigma_{h,i}^2}{K L_h} d_{uh,k,i} d_{bh,k,i} \right)\right. \\
		& \quad+ \left. \sum \limits_{i=1}^{N_{s1}} \log_2 \left( 1+ \frac{P_T}{\sigma^2} \frac{ N_u N_r^2 \sigma_{g,i}^2 \sigma_{t,i}^2}{K L_g L_t} d_{ut,k,i} d_{bg,k,i} d_{r,k,i} \right) \right]. \\
	\end{aligned}
\end{equation}
\end{lemma}
\begin{proof}
See Appendix \ref{append_pro2}.
\end{proof}

It is worth noting that the ergodic achievable rate approximations are derived for an arbitrary positive semidefinite Hermitian matrix $\mathbf{Q}$ {in our previous work \cite{li2022risassisted}}. However, in this work, the ergodic achievable rate approximations can be derived only when $\mathbf{Q}_k$ satisfies \eqref{approx_qk} or $\mathbf{Q}_k =  \frac{1}{KN_b} \mathbf{I}_{N_b}$ due to the existence of the BS-user  direct link.

\subsection{Tightness of the Approximations} \label{sec_tightness}
\begin{figure}[t]
\centering
\begin{minipage}[t]{0.4\textwidth}
\centering
\includegraphics[width=8cm]{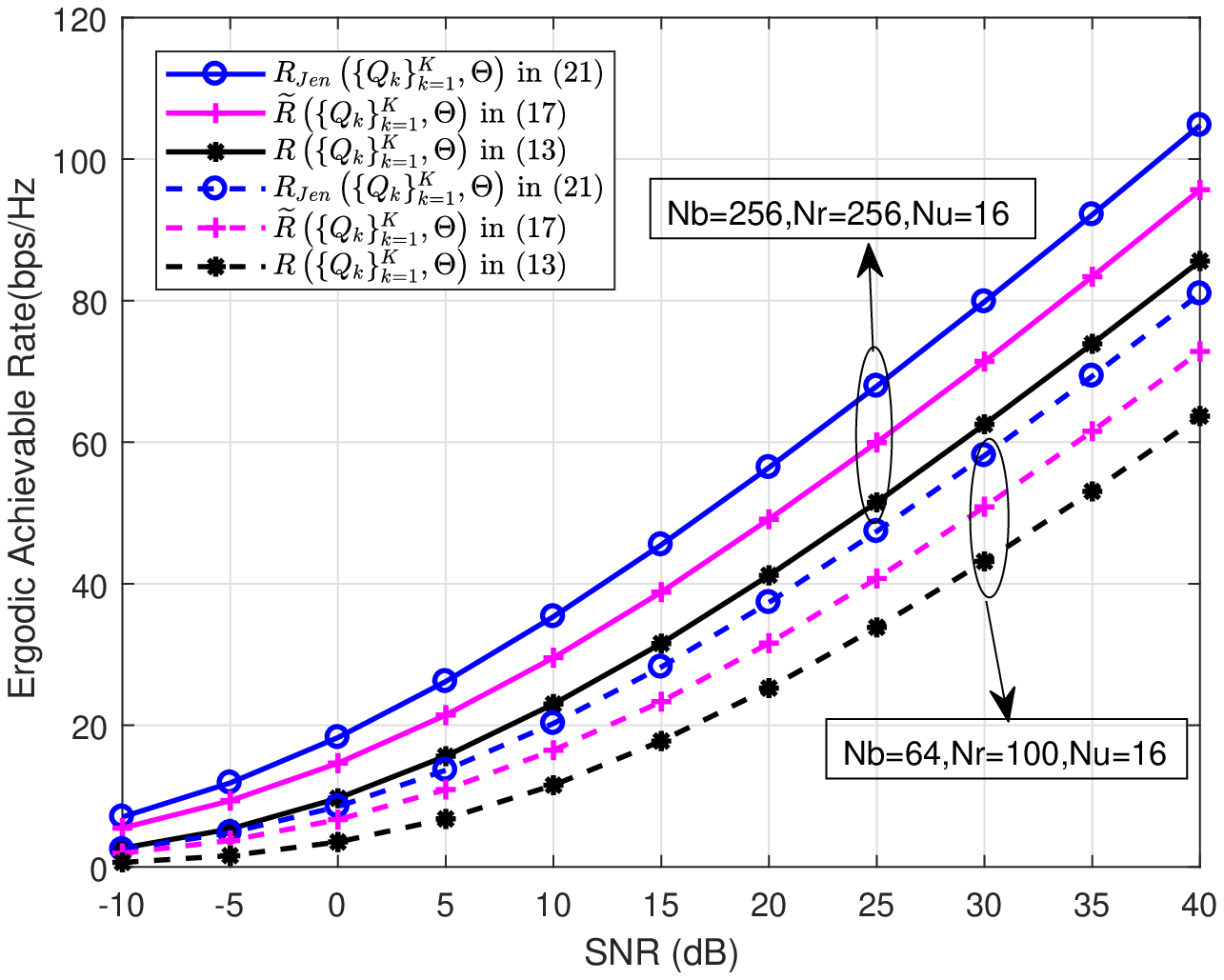}
\caption{{Ergodic achievable rate against SNR with different antenna numbers when $L_g=L_t=L_h=6$.}} \label{cap_num}
\end{minipage}
\begin{minipage}[t]{0.4\textwidth}
\centering
\includegraphics[width=8cm]{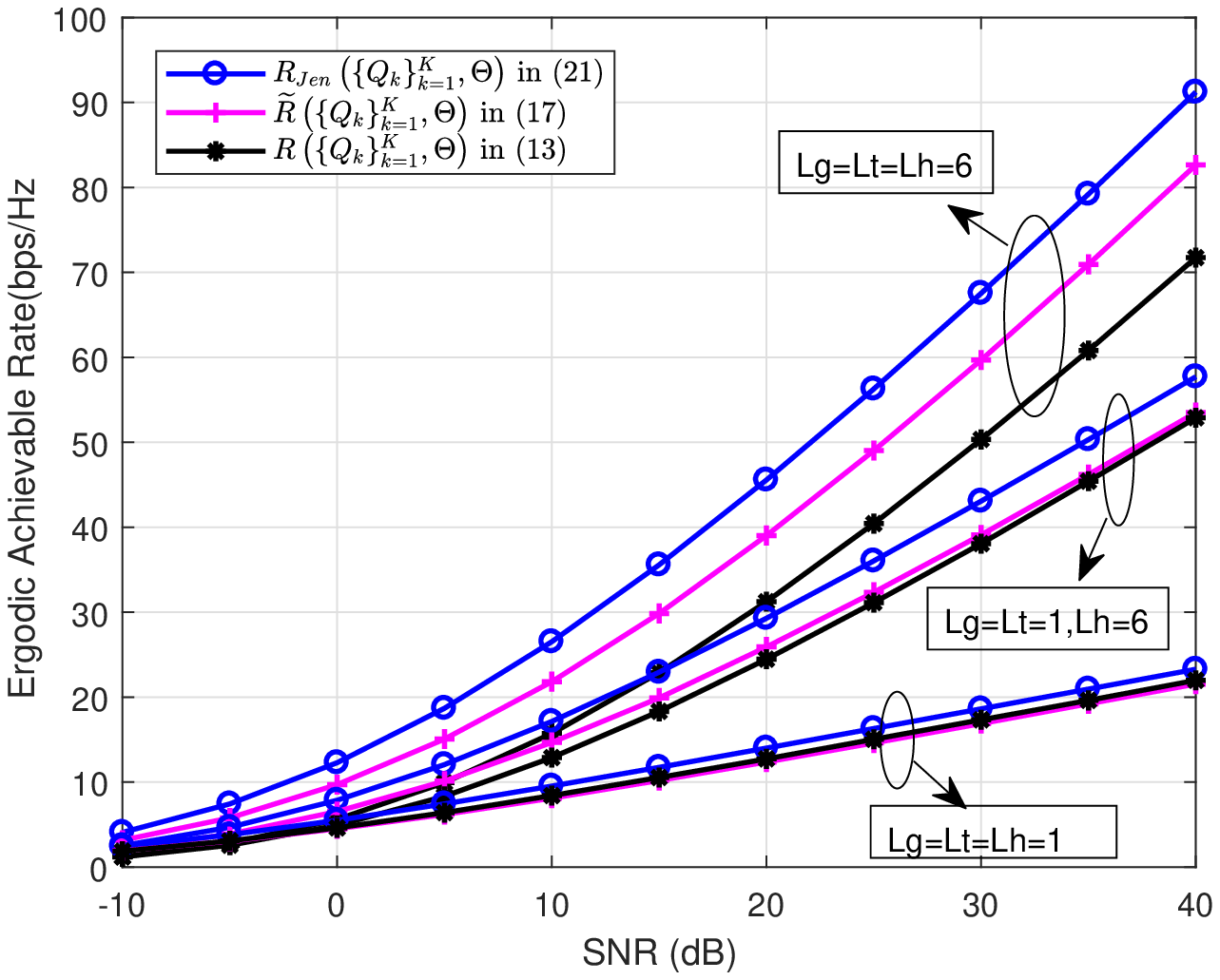}
\caption{{Ergodic achievable rate against SNR with different channel path when $N_b=100$, $N_r=169$ and $N_u=16$.}} \label{cap_path}
\end{minipage}
\end{figure}

\begin{figure}[t]
\begin{centering}
\includegraphics[width=.4\textwidth]{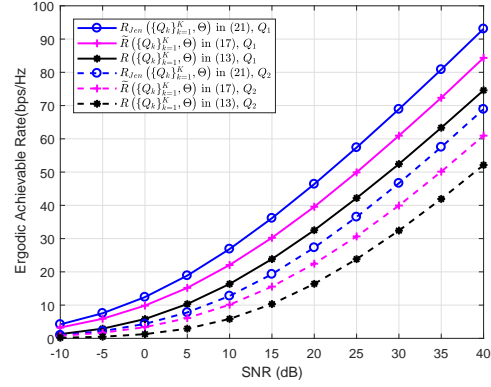}
 \caption{{Ergodic achievable rate against SNR with different $\mathbf{Q}_k$ when $L_g=L_t=L_h=6$, $N_b=100$, $N_r=169$, and $N_u=16$.}}\label{cap_q}
\end{centering}
\end{figure}

\begin{table*}[!t]
\centering
\caption{Normalized error of
$R_{Jen}\left(\{\mathbf{Q}_k\}_{k=1}^K, \mathbf{\Theta}\right)$ and $\widetilde{R}\left(\{\mathbf{Q}_k\}_{k=1}^K, \mathbf{\Theta}\right)$ with respect to ${R}(\{\mathbf{Q}_k\}_{k=1}^K, \mathbf{\Theta})$ when $\text{SNR}=20$ dB.} \label{compu_comp}
\begin{tabular}{|c|c|c|c|c|c|c|c|}
\hline  \multirow{2}{*}{} & \multicolumn{2}{|c|} {\text { Fig. 2 }} & \multicolumn{3}{|c|} {\text { Fig. 3 }} & \multicolumn{2}{|c|} {\text { Fig. 4 }} \\
\cline{2-8}&  $N_b=256,$  & $N_b=64,$ & $L_g=6,$ & $L_g=1,$ & $L_g=1,$& &\\
& $N_r=16\times 16,$  & $N_r=10\times 10,$ & $L_t=6,$ & $L_t=1,$ & $L_t=1,$ &$\mathbf{Q}_1$ & $\mathbf{Q}_2$ \\
& $N_u=16$  & $N_u=16$ & $L_h=6$ & $L_h=6$ & $L_h=1$ & &\\
\hline  {$ \frac{\left| R_{Jen}\left(\{\mathbf{Q}_k\}_{k=1}^K, \mathbf{\Theta}\right) - {R}\left(\{\mathbf{Q}_k\}_{k=1}^K, \mathbf{\Theta}\right)\right| } {{R}\left(\{\mathbf{Q}_k\}_{k=1}^K, \mathbf{\Theta}\right)}  $} &  37.0\%  & 47.9\% & 45.9\%& 19.7\% &  10.0\% & 42.6\% & 66.5\% \\
\hline $\frac{\left|\widetilde{R}\left(\{\mathbf{Q}_k\}_{k=1}^K, \mathbf{\Theta}\right) - {R}\left(\{\mathbf{Q}_k\}_{k=1}^K, \mathbf{\Theta}\right)\right|} {{R}\left(\{\mathbf{Q}_k\}_{k=1}^K, \mathbf{\Theta}\right)}$ &  19.3\%  & 25.2\% & 24.9\% &  5.9\% &  2.9\% &  21.7\% & 36.9\%\\
\hline
\end{tabular}
\end{table*}

{In this subsection, we investigate the tightness of the approximations via numerical results. Let us first consider the details of the channels as illustrated in Section \ref{channel}. The central carrier frequency is $f_c=28$ GHz, the number of subcarriers is $K=24$, the length of CP is $N_{CP}=10$, and the bandwidth is $f_s = 1$ GHz. The angles $\psi_{h,b,i}$, $\psi_{g,b,i}$, $\psi_{h,u,i}$ and $\psi_{t,u,i}$ are randomly generated and uniformly distributed in $[-\pi, \pi)$, $\phi_{g,r,i}$ and $\phi_{t,r,i}$ in $[-\pi/2,\pi/2]$, and $\varphi_{g,r,i}$ and $\varphi_{t,r,i}$ in $[0,\pi]$. The exponential distribution is adopted to model the variances of the complex gain \cite{7501500}. In order to make the complex gain of the direct link and the reflecting link comparable, we set $\sigma_{h,i}^2 \sim \exp (1)$, $\sigma_{g,i}^2 \sim \exp (0.1)$ and $\sigma_{t,i}^2 \sim \exp (0.1)$. All curves are averaged over 100 sets of independent realizations of the angles. For each set of angle realization, 1000 independent realizations of the path gains are generated for the Monte-Carlo simulations. Here, the phase matrix of RIS is randomly generated. We investigate the ergodic achievable rate against SNR with different numbers of antennas or reflection units (where $\mathbf{Q}_k$ is generated according to Eq. \eqref{approx_qk}, $\mathbf{\Lambda}_{q,k}= \frac{1}{K N_s} \mathbf{I}_{N_s}$, and $L_g=L_t=L_h=6$), different number of paths (where $\mathbf{Q}_k$ is generated according to Eq. \eqref{approx_qk}, $\mathbf{\Lambda}_{q,k}= \frac{1}{K N_s} \mathbf{I}_{N_s}$, $N_b=100$, $N_r=169$, and $N_u=16$), and different transmit covariance matrices (where $L_g=L_t=L_h=6$, $N_b=100$, $N_r=169$, $N_u=16$, $\mathbf{Q}_1$ is generated according to Eq. \eqref{approx_qk}, $\mathbf{\Lambda}_{q,k}= \frac{1}{K N_s} \mathbf{I}_{N_s}$, and $\mathbf{Q}_2=\frac{1}{KN_b} \mathbf{I}_{N_b}$) in Fig. \ref{cap_num}, \ref{cap_path}, and \ref{cap_q}, respectively. In addition, we  compare the normalized error of $R_{Jen}\left(\{\mathbf{Q}_k\}_{k=1}^K, \mathbf{\Theta}\right)$ and $\widetilde{R}\left(\{\mathbf{Q}_k\}_{k=1}^K, \mathbf{\Theta}\right)$ with respect to ${R}\left(\{\mathbf{Q}_k\}_{k=1}^K, \mathbf{\Theta}\right)$ when $\text{SNR}=20$ dB in Table \ref{compu_comp}. From Figs. \ref{cap_num}-\ref{cap_q} and Table \ref{compu_comp}, we find that the approximations $R_{Jen}\left(\{\mathbf{Q}_k\}_{k=1}^K, \mathbf{\Theta}\right)$ and $\widetilde{R}\left(\{\mathbf{Q}_k\}_{k=1}^K, \mathbf{\Theta}\right)$ are close to ${R}\left(\{\mathbf{Q}_k\}_{k=1}^K, \mathbf{\Theta}\right)$ when the number of paths is one, i.e., $L_g=L_t=L_h=1$. The approximations are larger than the actual results in other cases. It is mainly due to the following three facts. First of all, \textbf{Proposition} \ref{prop_orthogo} utilizes the asymptotic orthogonality of array response vectors. However, the number of antennas at the BS is usually limited in practice, which leads to  $\mathbf{G}_k \mathbf{H}_k^H \neq \mathbf{0}^{N_r\times N_u}$. Second, the optimal $\mathbf{Q}_k$ is closely related to the end-to-end instantaneous channel $\mathbf{H}_{\text{eff},k}$, but only the statistical CSI is available. Thus, the adopted transmit covariance matrix $\mathbf{Q}_k$ in Eq. \eqref{approx_qk} is sub-optimal, resulting in some errors. Third, $R_{Jen}\left(\{\mathbf{Q}_k\}_{k=1}^K, \mathbf{\Theta}\right)$ is amplified by the Jensen's inequality. Nevertheless, the trend of the derived approximations is consistent with the actual ergodic achievable rate, and it will be shown later that our proposed algorithm outperforms the algorithm in \cite{9234098}.}

\section{Optimization of Transmit Covariance Matrix and Reflection Coefficients} \label{sec_opt}
With the obtained  upper bound $R_{Jen}\left(\{\mathbf{Q}_k\}_{k=1}^K, \mathbf{\Theta}\right)$ of the ergodic achievable rate of the RIS-assisted mmWave MIMO-OFDM communication systems, we now concentrate on maximizing the ergodic achievable rate by jointly designing the transmit covariance matrix $\{\mathbf{Q}_k\}_{k=1}^K$ and the reflection coefficients $\mathbf{\Theta}$. Note that $\mathbf{Q}_k$ is determined by Eq. \eqref{approx_qk}. Thus, we only need to optimize $q_{k,i}$, instead of $\mathbf{Q}_k$. Applying $R_{Jen}\left(\{\mathbf{Q}_k\}_{k=1}^K, \mathbf{\Theta}\right)$ in Eq. \eqref{jen} to the problem $\mathcal{P}_0$, we have
\begin{subequations}\label{prob_p1}
\begin{align}
\mathcal{P}_1: \max \limits_{ q_{k,i}, \mathbf{\Theta} } \quad & R_{Jen}\left(\{\mathbf{Q}_k\}_{k=1}^K, \mathbf{\Theta}\right) \\
\text { s.t. } \quad & \sum \limits_{k=1}^K \sum \limits_{i=1}^{N_s}q_{k,i} \leq 1, \\
& q_{k,i} \geq 0, \forall i, k, \\
&  \mathbf{\Theta} =\operatorname{diag}\left(e^{j \theta_{1}}, e^{j \theta_{2}}, \cdots, e^{j \theta_{N_r}} \right),
\end{align}
\end{subequations}
where
\begin{equation}
	\begin{footnotesize}
	\begin{aligned}
		& R_{Jen}\left(\{\mathbf{Q}_k\}_{k=1}^K, \mathbf{\Theta}\right) \\
		& = \frac{1}{K+N_{cp}} \sum \limits_{k=1}^K \left[\sum \limits_{i=1}^{N_{s2}} \log_2 \left( 1+ \frac{P_T}{\sigma^2} \frac{N_b N_u \sigma_{h,i}^2}{L_h} q_{k,N_{s1}+i} d_{uh,k,i} d_{bh,k,i} \right) \right. \\
		& \quad \left.+  \sum \limits_{i=1}^{N_{s1}} \log_2 \left( 1+ \frac{P_T}{\sigma^2} \frac{N_b N_u N_r^2 \sigma_{g,i}^2 \sigma_{t,i}^2}{L_g L_t} q_{k,i} d_{ut,k,i} d_{bg,k,i} d_{r,k,i} \right) \right] ,
	\end{aligned}
\end{footnotesize}
\end{equation}
where $d_{ut,k,i}$, $d_{bg,k,i}$, $d_{r,k,i}$, $d_{uh,k,i}$, and $d_{bh,k,i}$ are defined in \textbf{Lemma} \ref{theorem_app}. Compared with problem $\mathcal{P}_0$, problem $\mathcal{P}_1$ is much simpler, where the objective function is only composed of the summation of simple logarithm functions. However, the problem $\mathcal{P}_1$ is still challenging mainly due to two facts. First of all, the problem $\mathcal{P}_1$ is not jointly convex over $q_{k,i}$ and $\mathbf{\Theta}$. Thus, the optimal solution is very difficult to obtain. Secondly, the reflection coefficients $\mathbf{\Theta}$ are not directly related to the objective function. We need to build the bridge to connect $\mathbf{\Theta}$ and the objective function. In the following, AO is adopted to decouple the problem $\mathcal{P}_1$ into two sub-problems, and then different tools are employed to optimize them.

\subsection{Transmit Covariance Matrix Optimization Given RIS Reflection Coefficients}
Since $\mathbf{Q}_k$ is determined by \eqref{approx_qk}, the problem of optimizing transmit covariance matrix can be transformed into a power allocation problem. Specifically, when the reflection coefficients $\mathbf{\Theta}$ is fixed, the problem $\mathcal{P}_1$ is simplified as follows,
\begin{subequations}\label{opt_q}
\begin{align}
\max \limits_{ q_{k,i} } \quad & \frac{1}{K+N_{cp}} \sum \limits_{k=1}^K \sum \limits_{i=1}^{N_{s}} \log_2 \left(1+ \zeta_{k,i} q_{k,i} \right) \\
\text { s.t. } \quad & \sum \limits_{k=1}^K \sum \limits_{i=1}^{N_s}q_{k,i} \leq 1, \\
& q_{k,i} \geq 0, \forall i, k,
\end{align}
\end{subequations}
where
\begin{equation}
  \zeta_{k,i} =
  \begin{cases}
  \frac{P_T}{\sigma^2} \frac{N_b N_u N_r^2 \sigma_{g,i}^2 \sigma_{t,i}^2}{L_g L_t} d_{ut,k,i} d_{bg,k,i} d_{r,k,i}, & \mbox{ if $ 1 \leq i \leq N_{s1}$,} \\
  \frac{P_T}{\sigma^2} \frac{N_b N_u \sigma_{h,i}^2}{L_h} d_{uh,k,i} d_{bh,k,i}, & \mbox{ if $ N_{s1}+1 \leq i \leq N_{s}$.}
  \end{cases}
\end{equation}
The above problem is well studied in \cite{1045248}. It is found that joint space-frequency water-filling outperforms single dimensional water-filling either in the spatial domain or in the frequency domain. The optimal power allocation $q_{k,i}$ across subcarriers and streams  is given by
\begin{equation} \label{ao_q}
  q_{k,i} = \max\left(0, \frac{1}{q_0} - \frac{1}{\zeta_{k,i}}\right),
\end{equation}
where $q_0$ is the water level satisfying $\sum \limits_{k=1}^K \sum \limits_{i=1}^{N_s} q_{k,i} = 1$.

\subsection{Reflection Coefficient Optimization Given Transmit Covariance Matrix}
When the power allocations $q_{k,i}$ are fixed, the problem $\mathcal{P}_1$ can be rewritten as (with constant terms ignored)
\begin{subequations}\label{opt_ris}
\begin{align}
 \min \limits_{ \mathbf{\Theta} } \quad & - \sum \limits_{k=1}^K \sum \limits_{i=1}^{N_{s1}} \log_2 (1+ \eta_{k,i} d_{r,k,i}) \\
\text { s.t. } \quad & \mathbf{\Theta} =\operatorname{diag}\left(e^{j \theta_{1}}, e^{j \theta_{2}}, \cdots, e^{j \theta_{N_r}} \right),
\end{align}
\end{subequations}
where $\eta_{k,i} = \frac{P_T}{\sigma^2} \frac{N_b N_u N_r^2 \sigma_{g,i}^2 \sigma_{t,i}^2}{L_g L_t} q_{k,i} d_{ut,k,i} d_{bg,k,i}$. For better understanding, the information of $d_{r,k,i}$ is rewritten down here: $d_{r,k,i}$ are the descending ordered eigenvalues of $\mathbf{X}_{r,k}^H \mathbf{X}_{r,k}$, where $\mathbf{X}_{r,k} = \mathbf{A}_{rt,k}^H \mathbf{\Theta} \mathbf{A}_{rg,k}$, $\mathbf{X}_{r,k}^H \mathbf{X}_{r,k} = \mathbf{U}_{r,k} \mathbf{D}_{r,k} \mathbf{U}_{r,k}^H$, $\mathbf{D}_{r,k} = \operatorname{diag} (d_{r,k,1}, d_{r,k,2}, \ldots, d_{r,k,L_g})$. The above problem is highly challenging due to the following two facts. On the one hand, problem \eqref{opt_ris} is non-convex due to the unit-modulus constraints, and thus difficult to solve optimally. On the other hand, the objective function is connected to the optimization variable via eigenvalue decomposition (EVD) rather than being directly tied to it. A sub-optimal solution of $\mathbf{\Theta}$ can be obtained using the RCG algorithm which is widely applied in RIS-aided systems \cite{8982186,9234098}.
Define $\boldsymbol{\theta} = \left[ e^{j\theta_1}, e^{j\theta_2}, \ldots, e^{j\theta_{N_r}} \right]^T \in \mathbb{C}^{N_r\times 1}$ and $f(\boldsymbol{\theta}) = - \sum \limits_{k=1}^K \sum \limits_{i=1}^{N_{s1}} \log_2 (1+ \eta_{k,i} d_{r,k,i})$. Then, the feasible set of $\boldsymbol{\theta}$ forms a Riemannian manifold $\mathcal{M}= \{\boldsymbol{\theta} \in \mathbb{C}^{N_r \times 1}: |\boldsymbol{\theta}_{i}|=1, \forall i\}$ \cite{AbsilMahonySepulchre+2009}. The RCG algorithm usually has three key steps in each iteration.

1) \textsl{Riemannian Gradient}: The Riemannian gradient $\operatorname{grad} f(\boldsymbol{\theta})$ is the orthogonal projection of the Euclidean gradient $\nabla  f(\boldsymbol{\theta})$ onto the manifold:
\begin{equation} \label{grad}
\operatorname{grad} f(\boldsymbol{\theta})=\nabla f(\boldsymbol{\theta}) -\Re \left\{\nabla f(\boldsymbol{\theta}) \odot \boldsymbol{\theta}^*\right\} \odot \boldsymbol{\theta}.
\end{equation}
In the following, we will explore the Euclidean gradient of $f(\boldsymbol{\theta})$.

For simplicity, let us first consider subcarrier $k$. Specifically, following the chain rule, the $\ell$-th element of the gradient is given by
\begin{equation} \label{daoshu}
\begin{aligned}
	&\left( \frac{\partial f(\boldsymbol{\theta})}{\partial \overline{\theta}_{\ell}^H} \right) _k \\
	&= \left[\begin{array}{llll}
	\frac{\partial d_{r,k,1}}{\partial \overline{\theta}_{\ell}^H} & \frac{\partial d_{r,k,2}}{\partial \overline{\theta}_{\ell}^H} & \cdots & \frac{\partial d_{r,k,N_{s1}}}{\partial \overline{\theta}_{\ell}^H}
	\end{array}\right]\left[\begin{array}{c}
	\frac{\partial f(\boldsymbol{\theta})}{\partial d_{r,k,1}} \\
	\frac{\partial f(\boldsymbol{\theta})}{\partial d_{r,k,2}} \\
	\vdots \\
	\frac{\partial f(\boldsymbol{\theta})}{\partial d_{r,k,N_{s1}}}
	\end{array}\right],
\end{aligned}
\end{equation}
where $\overline{\theta}_\ell$ denotes the $\ell$-th element of $\boldsymbol{\theta}$,  $\left( \frac{\partial f(\boldsymbol{\theta})}{\partial \overline{\theta}_{\ell}^H} \right) _k$ denotes the $\ell$-th element of the  gradient on subcarrier $k$, and the overall gradient can be represented as
 \begin{subequations} \label{daoshu0}
 \begin{align}
    &\nabla  f(\boldsymbol{\theta}) \notag \\
    &= \left[\frac{\partial f(\boldsymbol{\theta})}{\partial \overline{\theta}_1^H}, \frac{\partial f(\boldsymbol{\theta})}{\partial \overline{\theta}_2^H}, \ldots, \frac{\partial f(\boldsymbol{\theta})}{\partial \overline{\theta}_{N_r}^H}\right]^T \\
    &= \left[ \sum \limits_{k=1}^K \left( \frac{\partial f(\boldsymbol{\theta})}{\partial \overline{\theta}_{1}^H} \right) _k, \sum \limits_{k=1}^K \left( \frac{\partial f(\boldsymbol{\theta})}{\partial \overline{\theta}_{2}^H} \right) _k, \ldots, \sum \limits_{k=1}^K \left( \frac{\partial f(\boldsymbol{\theta})}{\partial \overline{\theta}_{N_r}^H} \right) _k \right]^T.
 \end{align}
\end{subequations}
The partial derivative of $f(\boldsymbol{\theta})$ with respect to $d_{r,k,i}$ can be found by directly differentiating $f(\boldsymbol{\theta})$ to be
\begin{equation} \label{daoshu1}
  \frac{\partial f(\boldsymbol{\theta})}{\partial d_{r,k,i}} = -\frac{\eta_{k,i}} { \left( 1+ \eta_{k,i}d_{r,k,i} \right) \ln 2}
\end{equation}
Then, we will explore the partial derivative of $d_{r,k,i}$ with respect to the $\ell$-th reflection coefficient. First of all, $d_{r,k,i}$ can be expressed as
\begin{subequations}
\begin{align}
  d_{r,k,i} & = \mathbf{u}_{k,i}^H \mathbf{X}_{r,k}^H \mathbf{X}_{r,k} \mathbf{u}_{k,i} \\
  & = \mathbf{u}_{k,i}^H \mathbf{A}_{rg,k}^H \mathbf{\Theta}^H \mathbf{A}_{rt,k}\mathbf{A}_{rt,k}^H \mathbf{\Theta} \mathbf{A}_{rg,k} \mathbf{u}_{k,i},
\end{align}
\end{subequations}
where $\mathbf{u}_{k,i}$ refers to the $i$-th column of $\mathbf{U}_{r,k}$. Then, the partial derivative with respect to $\overline{\theta}_\ell^H$ can be calculated as
\begin{subequations}\label{daoshu3}
\begin{align}
&\frac{\partial d_{r,k,i}}{\partial \overline{\theta}_{\ell}^H} \notag \\
&= \frac{\partial \mathbf{u}_{k,i}^H}{\partial \overline{\theta}_{\ell}^H} \mathbf{X}_{r,k}^H \mathbf{X}_{r,k} \mathbf{u}_{k,i} + \mathbf{u}_{k,i}^H \frac{\partial \left(\mathbf{X}_{r,k}^H \mathbf{X}_{r,k}\right)}{\partial \overline{\theta}_{\ell}^H} \mathbf{u}_{k,i} \notag\\
&\quad+ \mathbf{u}_{k,i}^H \mathbf{X}_{r,k}^H \mathbf{X}_{r,k} \frac{\partial \mathbf{u}_{k,i}}{\partial \overline{\theta}_{\ell}^H} \\
&\overset{(a)}{=}   d_{r,k,i} \frac{\partial \mathbf{u}_{k,i}^H}{\partial \overline{\theta}_{\ell}^H} \mathbf{u}_{k,i} + \mathbf{u}_{k,i}^H \frac{\partial \left(\mathbf{X}_{r,k}^H \mathbf{X}_{r,k}\right)}{\partial \overline{\theta}_{\ell}^H} \mathbf{u}_{k,i} + d_{r,k,i} \mathbf{u}_{k,i}^H \frac{\partial \mathbf{u}_{k,i}}{\partial \overline{\theta}_{\ell}^H} \\
&\overset{(b)}{=}   \mathbf{u}_{k,i}^H \frac{\partial \left(\mathbf{X}_{r,k}^H \mathbf{X}_{r,k}\right)}{\partial \overline{\theta}_{\ell}^H} \mathbf{u}_{k,i} \\
&= \mathbf{u}_{k,i}^H \frac{\partial \left(\mathbf{A}_{rg,k}^H \mathbf{\Theta}^H \mathbf{A}_{rt,k}\mathbf{A}_{rt,k}^H \mathbf{\Theta} \mathbf{A}_{rg,k}\right)}{\partial \overline{\theta}_{\ell}^H} \mathbf{u}_{k,i} ,
\end{align}
\end{subequations}
where (a) holds due to the fact that $\mathbf{X}_{r,k}^H \mathbf{X}_{r,k} \mathbf{u}_{k,i} = d_{r,k,i} \mathbf{u}_{k,i}$ and $\mathbf{u}_{k,i}^H \mathbf{X}_{r,k}^H \mathbf{X}_{r,k}= d_{r,k,i} \mathbf{u}_{k,i}^H$ by the definition of eigenvalues, and (b) follows $\mathbf{u}_{k,i}^H \mathbf{u}_{k,i} =1$.

The $(x,y)$-th entry of $\mathbf{X}_{r,k}^H \mathbf{X}_{r,k}$ can be represented as
\begin{equation}
\begin{small}	
\begin{aligned}	
  \left[\mathbf{X}_{r,k}^H \mathbf{X}_{r,k}\right]_{x,y} = \sum \limits_{m=1}^{N_r} \sum \limits_{n=1}^{N_r} \left[\mathbf{A}_{rg,k}^H \right]_{x,m} \overline{\theta}_m^H \left[ \mathbf{A}_{rt,k}\mathbf{A}_{rt,k}^H \right]_{m,n} \overline{\theta}_n \left[ \mathbf{A}_{rg,k}\right]_{n,y}  
\end{aligned}  
\end{small}  
\end{equation}
yielding the partial derivative with respect to $\overline{\theta}_\ell^H$
\begin{equation}
\begin{small}
\begin{aligned}
  \left[\frac{\partial \left(\mathbf{X}_{r,k}^H \mathbf{X}_{r,k}\right)}{\partial \overline{\theta}_{\ell}^H}\right]_{x,y} = \sum \limits_{n=1, n \neq \ell}^{N_r} \left[\mathbf{A}_{rg,k}^H \right]_{x,\ell} \left[ \mathbf{A}_{rt,k}\mathbf{A}_{rt,k}^H \right]_{\ell,n} \overline{\theta}_n \left[ \mathbf{A}_{rg,k}\right]_{n,y}
\end{aligned}  
\end{small}  
\end{equation}
Thus, the partial derivative of $\mathbf{X}_{r,k}^H \mathbf{X}_{r,k}$ with respect to $\overline{\theta}_\ell^H$ can be expressed as
\begin{equation} \label{daoshu2}
\begin{small}	
\begin{aligned}	
  \frac{\partial \left(\mathbf{X}_{r,k}^H \mathbf{X}_{r,k}\right)}{\partial \overline{\theta}_{\ell}^H} = \sum \limits_{n=1, n \neq \ell}^{N_r} \left[\mathbf{A}_{rg,k}^H \right]_{:,\ell} \otimes \left( \left[ \mathbf{A}_{rt,k}\mathbf{A}_{rt,k}^H \right]_{\ell,n} \overline{\theta}_n \left[ \mathbf{A}_{rg,k}\right]_{n,:} \right),
\end{aligned}  
\end{small}  
\end{equation}
where $\otimes$ denotes the Kronecker product operation.
Finally, by substituting \eqref{daoshu}, \eqref{daoshu1}, \eqref{daoshu3} and \eqref{daoshu2} back into \eqref{daoshu0}, we can obtain the Euclidean gradient of $f(\boldsymbol{\theta})$.
\begin{algorithm}[t]
    \caption{RCG Algorithm for Reflection Coefficients Optimization}
    \label{pgd}
    \begin{algorithmic}[1]
    \State {\bf Input:} $\{\eta_{k,i}, \mathbf{A}_{rt,k}, \mathbf{A}_{rg,k}\}$, desired accuracy $\epsilon$
    \State Initialize: $\boldsymbol{\theta}_{0}$, $\boldsymbol{\eta}_0 = - \operatorname{grad} f(\boldsymbol{\theta}_{0})$  and set $i=0$;
    \Repeat
            \State Choose the Armijo backtracking line search step size $\alpha$;
            \State  Update $\boldsymbol{\theta}$ by using \eqref{retraction};
            \State Update the Riemannian gradient by using \eqref{grad};
            \State Update the research direction by using \eqref{direction};
            \State $i \leftarrow i+1$;
    \Until $\|\operatorname{grad} f(\boldsymbol{\theta}_i) \|_2 \leq \epsilon$.
    \end{algorithmic}
 \end{algorithm}

2) \textsl{Search Direction}: The current search direction $\boldsymbol{\eta}$ can be found as
\begin{equation} \label{direction}
  \boldsymbol{\eta} = -\operatorname{grad} f(\boldsymbol{\theta}) + \beta \mathcal{T} (\overline{\boldsymbol{\eta}}),
\end{equation}
where $\beta$ is chosen as the Polak-Ribiere parameter to achieve fast convergence \cite{AbsilMahonySepulchre+2009}, $\overline{\boldsymbol{\eta}}$ is the previous search direction, and $\mathcal{T} (\cdot)$ is a transport operation defined as
\begin{equation}
  \mathcal{T}(\boldsymbol{\eta}) = \boldsymbol{\eta} - \Re \{ \boldsymbol{\eta}\odot \boldsymbol{\theta}^* \} \odot \boldsymbol{\theta}.
\end{equation}

3) \textsl{Retraction}: The retraction operation is to project the tangent vector back to the manifold:
\begin{equation} \label{retraction}
  \boldsymbol{\theta} \leftarrow \operatorname{unt} (\boldsymbol{\theta} + \alpha \boldsymbol{\eta}),
\end{equation}
where $\alpha$ is the Armijo backtracking line search step size \cite{AbsilMahonySepulchre+2009}.

The key steps are introduced above, and the consequent RCG algorithm  for reflection coefficients optimization is summarized in \textbf{Algorithm} \ref{pgd}, which is guaranteed to converge to a stationary point \cite{AbsilMahonySepulchre+2009}.

\section{Summary and Discussions}
\begin{algorithm}[t]
    \caption{Alternating Optimization Algorithm for Problem $\mathcal{P}_1$}
    \label{alg_overall}
    \begin{algorithmic}[1]
    \State Initialize: $q_{k,i}^{(0)}= \frac{1}{KN_s} $, $\mathbf{\Theta}^{(0)}$ is randomly generated where the phases $\{\theta_\ell\}_{\forall \ell}$ are uniformly and independently distributed in $[0,2\pi)$, error tolerance $\epsilon$, and set $i=0$;
    \Repeat
        \State Calculate $q_{k,i}^{(i+1)}$ according to \eqref{ao_q} with fixed $\mathbf{\Theta}^{(i)}$;
        \State Calculate $\mathbf{\Theta}^{(i+1)}$ based on \textbf{Algorithm} \ref{pgd} with fixed $q_{k,i}^{(i+1)}$;
        \State $i \leftarrow i+1$;
    \Until The increase of the objective value of the problem $\mathcal{P}_1$ is below the threshold $\epsilon$.
    \end{algorithmic}
 \end{algorithm}
In this section, we will first discuss the convergence and the complexity of the proposed algorithm, and then compare it with other state-of-the-art methods \cite{9234098, nuti2021spectral}. Note that the algorithm in \cite{9234098} that is originally designed for the narrowband systems is extended to the wideband systems here. In addition, it is worth noting that our proposed algorithm relies on the statistical CSI, while the algorithms in \cite{9234098, nuti2021spectral} depend on the instantaneous CSI.
\subsection{Converge and Computational Complexity Analysis}
In summary, the proposed algorithm is illustrated in \textbf{Algorithm} \ref{alg_overall}. Let us first consider the convergence of the proposed algorithm. We adopt the AO framework in \textbf{Algorithm} \ref{alg_overall}. First of all, the optimal power allocation is obtained via \eqref{ao_q} with fixed reflection coefficients. Then, the sub-optimal reflection coefficients are obtained via \textbf{Algorithm} \ref{pgd} with fixed power allocation. Therefore, the objective value of the problem $\mathcal{P}_1$ is non-decreasing over iterations and the proposed algorithm is guaranteed to converge.

Next, let us consider the complexity of the proposed \textbf{Algorithm} \ref{alg_overall}. Firstly, the complexity of the water-filling algorithm is dominated by the procedure of SVD. Note that $\mathbf{Q}_k$ is determined by \eqref{approx_qk}, and optimizing $\mathbf{Q}_k$ is transformed into power allocation in \eqref{ao_q} whose complexity can be neglected. Secondly, the complexity of reflection coefficient optimization is dominated by the calculation of the gradient in \eqref{daoshu0}. The complexity of calculating \eqref{daoshu3} is $\mathcal{O}( N_r L_g^2)$. The complexity of calculating \eqref{daoshu0} is $\mathcal{O}(K N_{s_1}  N_r^2 L_g^2 )$. And the complexity of calculating \eqref{approx_qk} is $\mathcal{O}(N_b (L_g^2 +L_h^2))$. Note that $\mathbf{Q}_k$ in \eqref{approx_qk} only needs to be calculated once. Thus, the overall complexity of the proposed \textbf{Algorithm} \ref{alg_overall} is on the order of $\mathcal{O}(I_1 I_2 K N_{s_1}  N_r^2 L_g^2 + KN_b (L_g^2 +L_h^2))$, where $I_1$ and $I_2$ denote the number of iterations of \textbf{Algorithm} \ref{pgd} and \textbf{Algorithm} \ref{alg_overall} required to converge, respectively.

\subsection{Computational Complexity Comparison}
\begin{table*}[!t]
\centering
\caption{Algorithm Comparison.}  \label{comp_alg}
\begin{tabular}{|c|c|c|}
\hline Algorithm & Dominant computational complexity & Type of CSI \\
\hline Proposed algorithm & $\mathcal{O}\left(I_1 I_2  KN_{s_1} N_r^2 L_g^2+ K N_b(L_g^2+L_h^2) \right)$ & statistical CSI \\
\hline PGA-based algorithm in \cite{nuti2021spectral} & $\mathcal{O}\left( KI \left(N_r N_b N_u+ N_b^2 N_u\right) \right)$ & instantaneous CSI \\
\hline T-SVD-BF in \cite{9234098} & $\mathcal{O}\left(K N_r L_{1}+ K N_{u}^{2}  N_{b}\right) $ & instantaneous CSI\\
\hline
\end{tabular}
\end{table*}
\begin{table*} [!t]
\centering
\caption{Computational Time Comparison.} \label{compu_time}
\begin{tabular}{|c|c|c|c|c|}
\hline \multirow{2}{*} {} & \multicolumn{4}{|c|} {\text { Running time (s) }} \\
\cline { 2 - 5 } & $N_r=10\times 5$ & $N_r=10\times 10$ & $N_r=10\times 20$ & $N_r=10\times 40$ \\
\hline \text { Proposed algorithm } & 1.88 & 7.47  & 27.67  & 107.80  \\
\hline \text { PGA-based algorithm in \cite{nuti2021spectral} } & 0.84 &1.37 & 3.04 & 15.19 \\
\hline \text { T-SVD-BF in \cite{9234098} } & 0.05  & 0.07  & 0.11  & 0.63   \\
\hline
\end{tabular}
\end{table*}
In this subsection, we compare the computational complexity with other state-of-the-art methods \cite{9234098, nuti2021spectral} that assume the knowledge of perfect instantaneous CSI.
\begin{itemize}
  \item  Projected gradient ascent (PGA)-based algorithm  \cite{nuti2021spectral}: The AO method is adopted, where PGA is utilized to optimize the reflection coefficients, and the spatial-frequency water-filling method is used to optimize the transmit covariance matrix.
  \item Truncated-SVD-based beamforming (T-SVD-BF) \cite{9234098}: A manifold-based algorithm is firstly adopted to optimize the reflection coefficients, followed by the water-filling algorithm to handle the active beamforming at the BS. Note that each stage only needs to be carried out once, and there is no need for an alternating process.
\end{itemize}
The complexity of PGA-based algorithm \cite{nuti2021spectral} which is designed for general channels is on the order of $\mathcal{O}\left(KI\left( N_r(N_u N_b+ N_u^2) + N_b N_u^2+ N_b^2 N_u +N_u^3  \right)\right)$, where $I$ denotes the number of iterations required to converge. The authors in \cite{9234098} propose the T-SVD-BF algorithm specifically for mmWave systems which utilizes the inherent sparse structure of mmWave channels. The complexity of T-SVD-BF \cite{9234098} is on the order of $\mathcal{O}(K N_r L_1 + K N_b N_u \min (N_b, N_u) )$, where $L_1$ denotes the number of iterations required to converge when optimizing the reflection coefficients. Considering the case $N_b > N_u$, the dominant computational complexity of the respective algorithms is summarized in Table \ref{comp_alg}. We further compare the running time for various value of $N_r$ in Table \ref{compu_time}. The system settings are the same as that in Section \ref{sec_tightness}, where $N_b=100$, $N_u=16$, and $L_g=L_t=L_h=6$. The simulations are carried out on a computer with Intel i7-7700 CPU at 3.60 GHz and with 16.0 GB RAM. We can find that the complexity of T-SVD-BF in  \cite{9234098} is the lowest, while our proposed algorithm has the highest complexity because we have to calculate the element-wise gradient in Eq. \eqref{daoshu0}.

\section{Simulation Results}
\begin{figure}[t]
\begin{centering}
\includegraphics[width=.4\textwidth]{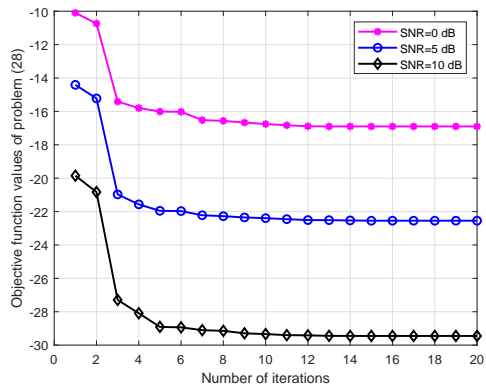}
 \caption{{Convergence of the proposed \textbf{Algorithm} \ref{pgd}.}}\label{fig_conv_rcg}
\end{centering}
\end{figure}

\begin{figure}[t]
\begin{centering}
\includegraphics[width=.4\textwidth]{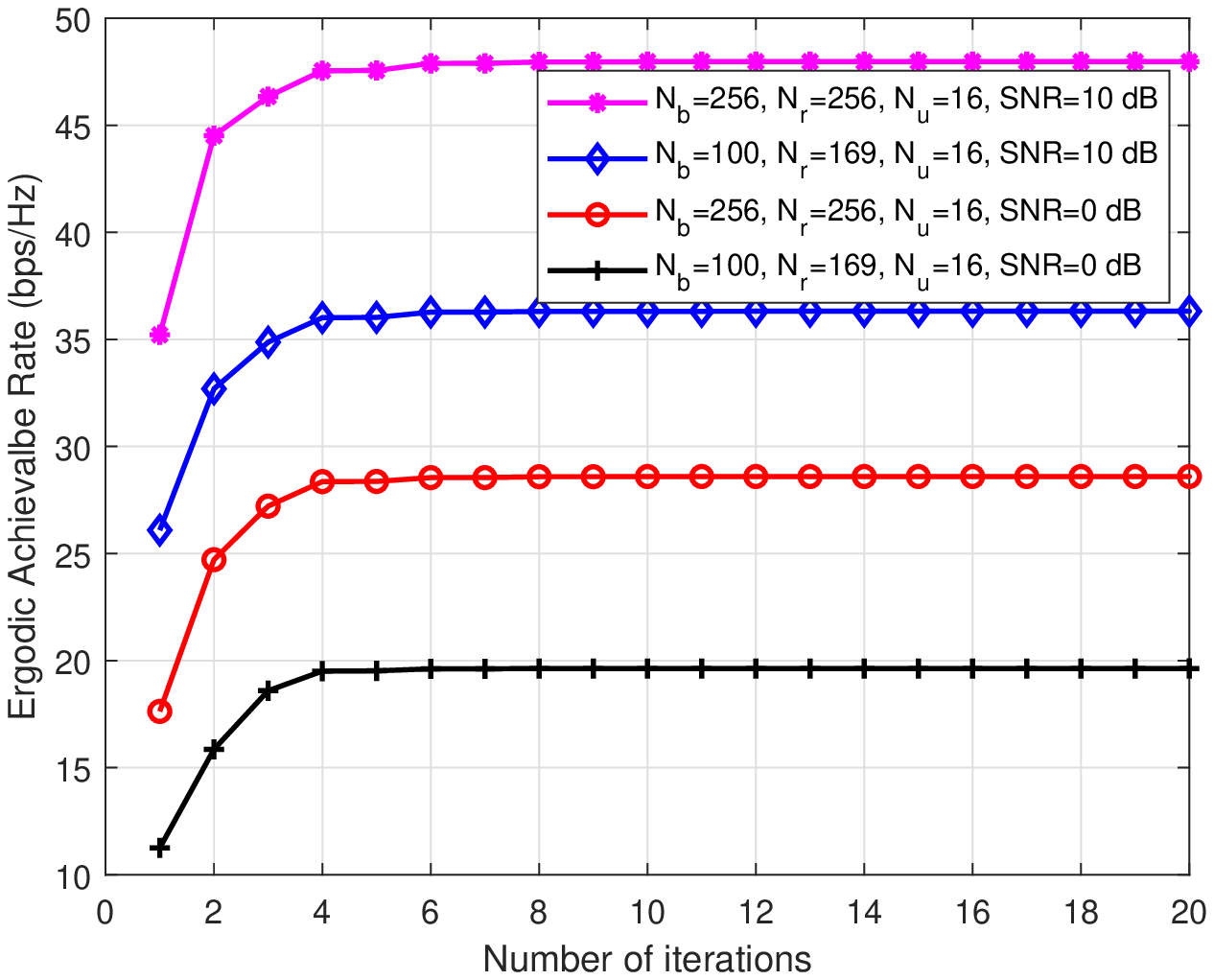}
 \caption{{Convergence of the proposed \textbf{Algorithm} \ref{alg_overall}.}}\label{fig_conv}
\end{centering}
\end{figure}

\begin{figure}[t]
\begin{centering}
\includegraphics[width=.4\textwidth]{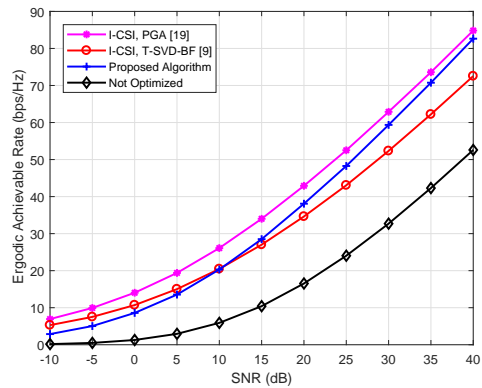}
 \caption{{Performance comparison with benchmarks, I-CSI denotes instantaneous CSI.}}\label{fig_bench}
\end{centering}
\end{figure}
In this section, we evaluate the performance of the proposed algorithms for ergodic achievable rate optimization. The system settings are the same as that in Section \ref{sec_tightness}. Other system parameters are set as follows unless specified otherwise later: $N_b=100$, $N_r = 13 \times 13$, $N_u=16$, and $L_h=L_g=L_t=6$. All curves are averaged over 100 sets of independent realizations of the angles. For each set of angle realization, 1000 independent realizations of the path gains are generated for the Monte-Carlo simulations.

Fig.~\ref{fig_conv_rcg} illustrates the convergence performance of the proposed \textbf{Algorithm} \ref{pgd} where $q_{k,i}$ is fixed at $\frac{1}{KN_s}$ in different SNR scenarios. We can find that \textbf{Algorithm} \ref{pgd} converges in around 10 iterations.
Fig.~\ref{fig_conv} displays the convergence performance of the proposed \textbf{Algorithm} \ref{alg_overall}. It is found that \textbf{Algorithm} \ref{alg_overall} converges in about 6 iterations, which confirms the convergence of the proposed algorithm.

In Fig.~\ref{fig_bench}, the proposed algorithm is compared with some benchmarks, where PGA and T-SVD-BF refer to the algorithm in \cite{nuti2021spectral} and \cite{9234098}, respectively, and ``Not Optimized'' denotes the scenario where $q_{k,i} = \frac{1}{K N_s}$ and $\mathbf{\Theta}$ is randomly generated. It is interesting that our proposed algorithm performs well although the derived approximation $R_{Jen}\left(\{\mathbf{Q}_k\}_{k=1}^K, \mathbf{\Theta}\right)$ of the ergodic achievable rate  is larger than the Monte-Carlo results. On the one hand, the PGA-based algorithm in \cite{nuti2021spectral} performs best. The gap between our proposed algorithm and the PGA-based algorithm narrows as SNR increases. When SNR is 40 dB, our proposed algorithm is about 2 bps/Hz worse than the PGA-based algorithm. On the other hand, the proposed algorithm outperforms T-SVD-BF when SNR is larger than 10 dB. Both T-SVD-BF and our proposed algorithm utilize the asymptotic orthogonality of array response vectors. The T-SVD-BF algorithm requires both BS and user sides to employ a large number of antennas to approximate the SVD of the BS-RIS-user reflection channel. However, our proposed algorithm only requires the BS side to employ  a large number of antennas.  Nevertheless, the complexity of T-SVD-BF is the least and all the three algorithms improve the ergodic achievable rate greatly compared with the curve ``Not Optimized''. Note that our proposed algorithm only requires statistical  CSI, while the PGA-based algorithm and T-SVD-BF require instantaneous CSI. Overall, our proposed algorithm is quite favorable in practice considering the performance and CSI feedback overhead.


\begin{figure}[t]
\begin{centering}
\includegraphics[width=.4\textwidth]{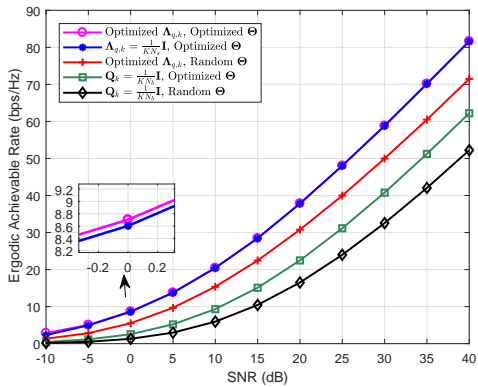}
 \caption{{Influence comparison between transmit covariance matrix and reflection coefficients.}}\label{fig_comp_monte}
\end{centering}
\vspace{-0.8cm}
\end{figure}


We compare the influence on the ergodic achievable rate of the transmit covariance matrix and reflection coefficients in Fig.~\ref{fig_comp_monte}. When $\mathbf{Q}_k = \frac{1}{KN_b}\mathbf{I}$, the Jensen's approximation $R_{Jen} \left(\mathbf{\Theta} \bigg| \mathbf{Q}_k= \frac{1}{KN_b} \mathbf{I}_{N_b} \right)$ in Eq. \eqref{Qid} is adopted in the optimization.  We first find that the ergodic achievable rate of equal power allocation approaches that of the optimized power allocation, especially when the SNR is greater than 10 dB. Then, compared with the curve ``$\mathbf{Q}_k = \frac{1}{KN_b}\mathbf{I}$, Random $\mathbf{\Theta}$'', the ergodic achievable rate can be improved by about 8 bps/Hz if only the reflection coefficients $\mathbf{\Theta}$ are optimized when SNR is larger than 30 dB. The ergodic achievable rate can be improved by about 16 bps/Hz if only the transmit covariance matrix is optimized, while the ergodic achievable rate can be improved by about 25 bps/Hz if both the reflection coefficients and the transmit covariance matrix are optimized. These observations indicate that the design of the reflection coefficients and the transmit covariance matrix plays a crucial role for ergodic achievable rate.

\section{Conclusion}
This paper studies the ergodic achievable rate maximization problem for a point-to-point RIS-assisted mmWave MIMO-OFDM communication system. First of all,  compact closed-form approximations of the ergodic achievable rate are derived by means of the majorization theory and Jensen's inequality. The approximations show that the ergodic achievable rate increases logarithmically with  the number of antennas at the BS and the user, the number of the reflection units at the RIS, the power allocation at the BS, as well as the eigenvalues of the array-response-related matrices associated with the BS, the user, and the RIS. Then, an AO-based algorithm is proposed to maximize the ergodic achievable rate by jointly optimizing the transmit covariance matrix at the BS and the reflection coefficients at the RIS, where the transmit covariance matrix is optimized by spatial-frequency water-filling and the reflection coefficients are optimized by the RCG algorithm. Simulation results validate  the effectiveness of the proposed algorithms. The ergodic achievable rate after optimization can be improved by about 25 bps/Hz on average, which demonstrates the effectiveness of the optimization of the transmit covariance matrix and the reflection coefficients.

\appendix
\begin{appendices}
\subsection{Proof of Proposition 1} \label{append_pro1}
The AoDs of different scattering paths can be considered as continuous random variables that are independent from each other. Then, the event $E=\{\psi_{g,b,i} \neq \psi_{h,b,j}, \forall i\in \{1,2,\ldots,L_g\}, \linebreak[4]\forall j \in \{1,2,\ldots, L_h\} \}$ occurs with probability one. Therefore, by the asymptotic orthogonality of ULA array response vectors \cite{9234098, 6831723}, we have
\begin{equation}
  \mathbf{A}_{bg,k}^H \mathbf{A}_{bh,k} \rightarrow \mathbf{0}^{L_g \times L_h}, \quad \text{as}\; N_b \rightarrow \infty.
\end{equation}
Consequently, $ \mathbf{G}_k \mathbf{H}_k^H = \mathbf{A}_{rg,k} \mathbf{G}_{L,k} \mathbf{A}_{bg,k}^H \mathbf{A}_{bh,k} \mathbf{H}_{L,k}^H \mathbf{A}_{uh,k}^H \rightarrow \mathbf{0}^{N_r \times N_u} $.

\subsection{Proof of Lemma 1} \label{append_theorem1}
Let us first focus on subcarrier $k$, and then extend the results to all of the subcarriers. Define the SVD of $\mathbf{H}_{\text{eff},k}$ as $\mathbf{H}_{\text{eff},k} = \mathbf{U}\boldsymbol{\Sigma}\mathbf{V}^H$. If the instantaneous CSI $\mathbf{H}_{\text{eff},k}$ is known, the optimal transmit covariance matrix $\mathbf{Q}_k$ is given by
\begin{equation}
  \mathbf{Q}_k = \mathbf{V}_1 \boldsymbol{\Lambda}_k \mathbf{V}_1^H,
\end{equation}
where $\mathbf{V}_1$ represents the first $N_s$ columns of $\mathbf{V}$, $N_s$ represents the rank of $\mathbf{H}_{\text{eff},k}$, $\boldsymbol{\Lambda}_k = \operatorname{diag}(q_{k,1}, q_{k,2}, \ldots, q_{k,N_s})$ where $q_{k,i}$ represents the optimal amount of power allocated to the $i$-th data stream. Then, we have
\begin{subequations}
\begin{align}
R_k & = \mathbb{E}_{\mathbf{H}_{\text{eff},k}} \left[\log_2 \operatorname{det} \left( \mathbf{I}_{N_u}+ \frac{P_T}{\sigma^2} \mathbf{H}_{\text{eff},k} \mathbf{Q}_k \mathbf{H}_{\text{eff},k}^H \right)\right] \\
&= \mathbb{E}_{\mathbf{H}_{\text{eff},k}}\left[\sum \limits_{i=1}^{N_s} \log _{2} \left( 1+ \frac{P_T}{\sigma^{2}} q_{k,i} |\boldsymbol{\Sigma}(i,i)|^2 \right)\right]\\
  &= \mathbb{E}_{\mathbf{H}_{\text{eff},k}}\left[\sum \limits_{i=1}^{N_s} \log _{2} \left( 1+ \frac{P_T}{\sigma^{2}} q_{k,i} \lambda_i\left(\mathbf{H}_{\text{eff},k} \mathbf{H}_{\text{eff},k}^H\right) \right)\right]\\
  &\overset{(a)}{\approx} \mathbb{E}_{\mathbf{H}_{\text{eff},k}}\left[\sum \limits_{i=1}^{N_s} \log _{2} \left( 1+
\frac{P_T}{\sigma^{2}} q_{k,i} \lambda_i\left( \mathbf{T}_k \mathbf{\Theta}\mathbf{G}_k \mathbf{G}_k^H \mathbf{\Theta}^H \mathbf{T}_k^H \right.\right.\right. \notag\\
& \quad\quad\quad \left. + \mathbf{H}_k \mathbf{H}_k^H \right) \Big) \bigg]\\
&\overset{(b)}{\approx} \mathbb{E}_{\mathbf{H}_{\text{eff},k}}\left[ \sum \limits_{i=1}^{N_{s1}} \log _{2} \left( 1+
\frac{P_T}{\sigma^{2}} q_{k,i} \lambda_i\left( \mathbf{T}_k \mathbf{\Theta} \mathbf{G}_k \mathbf{G}_k^H \mathbf{\Theta}^H \mathbf{T}_k^H \right) \right) \right. \notag \\
& \quad\quad\quad  \left. + \sum \limits_{i=1}^{N_{s2}} \log _{2} \left( 1+
\frac{P_T}{\sigma^{2}} q_{k,i +N_{s_1}}  \lambda_i\left( \mathbf{H}_k \mathbf{H}_k^H \right) \right) \right],
\end{align}
\end{subequations}
where $R_k$ refers to the ergodic achievable rate at subcarrier $k$, $\lambda_i(\cdot)$ refers to the $i$-th largest eigenvalue of the input matrix, $(a)$ holds due to the \textbf{Proposition} \ref{prop_orthogo}, $(b)$ holds due to the asymptotic orthogonality of $\mathbf{G}_k$ and $\mathbf{H}_k^H$, $N_{s1}$ is the rank of the BS-RIS-user cascade channel, $N_{s2}$ is the rank of the direct BS-user channel, and it follows that $N_s = N_{s1}+N_{s2}$.

Define $\mathbf{X}_{r,k} = \mathbf{A}_{rt,k}^H \mathbf{\Theta} \mathbf{A}_{rg,k}$, and assume the EVD of $\mathbf{A}_{ut,k}^H \mathbf{A}_{ut,k}$, $\mathbf{A}_{bg,k}^H \mathbf{A}_{bg,k}$, $\mathbf{X}_{r,k}^H \mathbf{X}_{r,k}$, $\mathbf{A}_{uh,k}^H \mathbf{A}_{uh,k}$ and $\mathbf{A}_{bh,k}^H \mathbf{A}_{bh,k}$ can be expressed as follows,
\begin{subequations}
\begin{align}
 & \mathbf{A}_{ut,k}^H \mathbf{A}_{ut,k} = \mathbf{U}_{ut,k}^H \mathbf{D}_{ut,k} \mathbf{U}_{ut,k}, \\
& \mathbf{A}_{bg,k}^H \mathbf{A}_{bg,k} = \mathbf{U}_{bg,k}^H \mathbf{D}_{bg,k} \mathbf{U}_{bg,k}, \\
& \mathbf{X}_{r,k}^H \mathbf{X}_{r,k} = \mathbf{U}_{r,k}^H \mathbf{D}_{r,k} \mathbf{U}_{r,k}, \\
& \mathbf{A}_{uh,k}^H \mathbf{A}_{uh,k} = \mathbf{U}_{uh,k}^H \mathbf{D}_{uh,k} \mathbf{U}_{uh,k}, \\
& \mathbf{A}_{bh,k}^H \mathbf{A}_{bh,k} = \mathbf{U}_{bh,k}^H \mathbf{D}_{bh,k} \mathbf{U}_{bh,k},
 \end{align}
\end{subequations}
where $\mathbf{U}_{ut,k}$, $\mathbf{U}_{bg,k}$, $\mathbf{U}_{r,k}$, $\mathbf{U}_{uh,k}$, and $\mathbf{U}_{bh,k}$ are the eigenvectors of $\mathbf{A}_{ut,k}^H \mathbf{A}_{ut,k}$, $\mathbf{A}_{bg,k}^H \mathbf{A}_{bg,k}$, $\mathbf{X}_{r,k}^H \mathbf{X}_{r,k}$, $\mathbf{A}_{uh,k}^H \mathbf{A}_{uh,k}$ and $\mathbf{A}_{bh,k}^H \mathbf{A}_{bh,k}$, respectively; $ \mathbf{D}_{ut,k} = \operatorname{diag} (d_{ut,k,1}, d_{ut,k,2}, \ldots, d_{ut,k,L_t})$, $\mathbf{D}_{bg,k} = \operatorname{diag} (d_{bg,k,1}, d_{bg,k,2}, \ldots, d_{bg,k,L_g})$, $\mathbf{D}_{r,k} = \operatorname{diag} (d_{r,k,1}, d_{r,k,2}, \ldots, d_{r,k,L_g})$, $\mathbf{D}_{uh,k} =  \operatorname{diag} (d_{uh,k,1}, d_{uh,k,2}, \ldots, d_{uh,k,L_h})$, $\mathbf{D}_{bh,k} = \operatorname{diag} (d_{bh,k,1}, d_{bh,k,2}, \ldots, d_{bh,k,L_h})$, and $d_{ut,k,i}, d_{bg,k,i}, d_{r,k,i},  d_{uh,k,i}, d_{bh,k,i} \geq 0$  are the eigenvalue of $\mathbf{A}_{ut,k}^H \mathbf{A}_{ut,k}$, $\mathbf{A}_{bg,k}^H \mathbf{A}_{bg,k}$, $\mathbf{X}_{r,k}^H \mathbf{X}_{r,k}$, $\mathbf{A}_{uh,k}^H \mathbf{A}_{uh,k}$ and $\mathbf{A}_{bh,k}^H \mathbf{A}_{bh,k}$ in descending order, respectively.

{Using the results in \cite[Lemma 1]{li2022risassisted}} and \cite[Theorem 1]{8816689}, we have
\begin{equation} \label{ergodic_app_subk}
	\begin{aligned}
		R_{app,k} &= \mathbb{E}_{\lambda} \left[ \log_2 \operatorname{det} \left( \mathbf{I}_{N_{s2}}+ \frac{P_T}{\sigma^2} \boldsymbol{\Lambda}^2_k \odot \boldsymbol{\lambda}(\mathbf{A}_{uh,k}^H \mathbf{A}_{uh,k}) \right.\right. \\
		& \quad  \odot \boldsymbol{\lambda} (\mathbf{A}_{bh,k}^H \mathbf{A}_{bh,k}) \odot \boldsymbol{\lambda} (\mathbf{H}_{L,k}^H \mathbf{H}_{L,k})   \Big) \bigg]\\
		&\quad +\mathbb{E}_{\lambda} \left[ \log_2 \operatorname{det} \left( \mathbf{I}_{N_{s1}}+ \frac{P_T}{\sigma^2} \boldsymbol{\Lambda}^1_k \odot \boldsymbol{\lambda}(\mathbf{A}_{ut,k}^H \mathbf{A}_{ut,k}) \right. \right. \\
		& \quad \odot \boldsymbol{\lambda}(\mathbf{A}_{bg,k}^H \mathbf{A}_{bg,k}) \odot \boldsymbol{\lambda}(\mathbf{X}_{r,k}^H \mathbf{X}_{r,k}) \odot \boldsymbol{\lambda}(\mathbf{T}_{L,k}^H \mathbf{T}_{L,k}) \\
		& \quad \odot \boldsymbol{\lambda} (\mathbf{G}_{L,k}^H \mathbf{G}_{L,k}) \Big) \bigg] \\
		&= \mathbb{E}_{\alpha_{g}, \alpha_{t}} \left[\sum \limits_{i=1}^{N_{s1}} \log_2 \left( 1+ \frac{P_T}{\sigma^2} \frac{N_b N_u N_r^2}{L_g L_t} q_{k,i} d_{ut,k,i} \right.\right. \\
		& \quad \times d_{bg,k,i} d_{r,k,i} |\alpha_{g,i}|^2 |\alpha_{t,i}|^2 \Big) \bigg] \\
		&\quad  + \mathbb{E}_{\alpha_{h}} \left[\sum \limits_{i=1}^{N_{s2}} \log_2 \left( 1+ \frac{P_T}{\sigma^2} \frac{N_b N_u}{L_h} q_{k, N_{s1}+i} d_{uh,k,i} \right.\right. \\
		& \quad \times d_{bh,k,i} |\alpha_{h,i}|^2 \Big) \bigg]
	\end{aligned}
\end{equation}
where $R_{app,k}$ denotes the approximation of $R_k$, $\boldsymbol{\Lambda}^1_k = \operatorname{diag} (q_{1,k},q_{2,k}, \ldots, q_{N_{s1},k})$, $\boldsymbol{\Lambda}^2_k = \linebreak[4] \operatorname{diag} (q_{N_{s1}+1}, q_{N_{s1}+2}, \ldots, q_{Ns})$ and $\boldsymbol{\lambda}(\mathbf{Y}) = [\lambda_{1}(\mathbf{Y}),
\lambda_2 (\mathbf{Y}), \ldots, \lambda_{n}(\mathbf{Y})]^T$. {According to \cite[Corollary 1]{li2022risassisted} and \cite[Corollary 2]{8816689}, one condition of $R_{app,k}=R_k$ is that all steering vectors of $\mathbf{A}_{uh,k}$, $\mathbf{A}_{ut,k}$, $\mathbf{A}_{rt,k}$, $\mathbf{A}_{rg,k}$, $\mathbf{A}_{bg,k}$, and $\mathbf{A}_{bh,k}$ are composed of the columns of unitary matrices, e.g., the discrete Fourier matrices. However, all steering vectors are dependent of frequency. Therefore, $R_{app,k}, \forall k$ is never equal to $R_k, \forall k$ in wideband systems.}

{The result in \eqref{ergodic_app_subk} is only for one subcarrier.} Then, considering all of the subcarriers, the ergodic achievable rate can be approximated by Eq. \eqref{r_appx}.
\begin{figure*}[htb]
\begin{equation} \label{r_appx}
\begin{aligned}
\widetilde{R}\left(\{\mathbf{Q}_k\}_{k=1}^K, \mathbf{\Theta}\right)
 = & \; \mathbb{E}_{\alpha_{g}, \alpha_{t}} \left[ \frac{1}{K+N_{cp}} \sum \limits_{k=1}^K \sum \limits_{i=1}^{N_{s1}} \log_2 \left( 1+ \frac{P_T}{\sigma^2} \frac{N_b N_u N_r^2}{L_g L_t} q_{k,i} d_{ut,k,i} d_{bg,k,i} d_{r,k,i} |\alpha_{g,i}|^2 |\alpha_{t,i}|^2 \right) \right] \\
 & + \mathbb{E}_{\alpha_{h}} \left[ \frac{1}{K+N_{cp}} \sum \limits_{k=1}^K \sum \limits_{i=1}^{N_{s2}} \log_2 \left( 1+ \frac{P_T}{\sigma^2} \frac{N_b N_u}{L_h} q_{k,N_{s1}+i} d_{uh,k,i} d_{bh,k,i} |\alpha_{h,i}|^2 \right) \right].
\end{aligned}
\end{equation}
\end{figure*}

Note that the above results are based on the knowledge of the instantaneous CSI of $\mathbf{H}_{\text{eff},k}$ and the optimal design of the transmit covariance matrix $\mathbf{Q}_k$. However, only the statistical CSI is available in many practical systems and hence the assumption in this work. Meanwhile, we find that the optimization of the transmit covariance matrix is strongly related to the steering matrix at the BS. Therefore, one reasonable and sub-optimal design of the transmit covariance matrix $\mathbf{Q}_k$ is provided in Eq. \eqref{approx_qk}.

\subsection{Proof of Theorem 1}\label{append_theorem3}
Applying Jensen's inequality $\mathbb{E}\{\log_2 (1+x)\} \leq \log_2(1+\mathbb{E}\{x\})$ for $x\geq 0$ to \eqref{ergo_app}, we have Eq. \eqref{proof_theorem1},
\begin{figure*}[htb]
\begin{equation} \label{proof_theorem1}
\begin{aligned}
\widetilde{R}\left(\{\mathbf{Q}_k\}_{k=1}^K, \mathbf{\Theta}\right) & = \mathbb{E}_{\alpha_{g}, \alpha_{t}} \left[ \frac{1}{K+N_{cp}} \sum \limits_{k=1}^K \sum \limits_{i=1}^{N_{s1}} \log_2 \left( 1+ \frac{P_T}{\sigma^2} \frac{N_b N_u N_r^2}{L_g L_t} q_{k,i} d_{ut,k,i} d_{bg,k,i} d_{r,k,i} |\alpha_{g,i}|^2 |\alpha_{t,i}|^2 \right) \right] \\
 & \quad + \mathbb{E}_{\alpha_{h}} \left[ \frac{1}{K+N_{cp}} \sum \limits_{k=1}^K \sum \limits_{i=1}^{N_{s2}} \log_2 \left( 1+ \frac{P_T}{\sigma^2} \frac{N_b N_u}{L_h} q_{k,N_{s1}+i} d_{uh,k,i} d_{bh,k,i} |\alpha_{h,i}|^2 \right) \right]\\
& \leq  \frac{1}{K+N_{cp}} \sum \limits_{k=1}^K \sum \limits_{i=1}^{N_{s1}} \log_2 \left( 1+ \frac{P_T}{\sigma^2} \frac{N_b N_u N_r^2}{L_g L_t} q_{k,i} d_{ut,k,i} d_{bg,k,i} d_{r,k,i} \mathbb{E} \left(|\alpha_{g,i}|^2 |\alpha_{t,i}|^2\right) \right) \\
 & \quad + \frac{1}{K+N_{cp}} \sum \limits_{k=1}^K \sum \limits_{i=1}^{N_{s2}} \log_2 \left( 1+ \frac{P_T}{\sigma^2} \frac{N_b N_u}{L_h} q_{k,N_{s1}+i} d_{uh,k,i} d_{bh,k,i} \mathbb{E}\left(|\alpha_{h,i}|^2\right) \right) \\
& \overset{(a)}{=} \frac{1}{K+N_{cp}} \sum \limits_{k=1}^K \sum \limits_{i=1}^{N_{s1}} \log_2 \left( 1+ \frac{P_T}{\sigma^2} \frac{N_b N_u N_r^2 \sigma_{g,i}^2 \sigma_{t,i}^2}{L_g L_t} q_{k,i} d_{ut,k,i} d_{bg,k,i} d_{r,k,i} \right) \\
 & \quad + \frac{1}{K+N_{cp}} \sum \limits_{k=1}^K \sum \limits_{i=1}^{N_{s2}} \log_2 \left( 1+ \frac{P_T}{\sigma^2} \frac{N_b N_u \sigma_{h,i}^2}{L_h} q_{k,N_{s1}+i} d_{uh,k,i} d_{bh,k,i} \right) \\
& \triangleq {R}_{Jen}\left(\{\mathbf{Q}_k\}_{k=1}^K, \mathbf{\Theta}\right).
\end{aligned}
\end{equation}
\end{figure*}
where $(a)$ holds due to the following fact. $\alpha_{g,i} \sim \mathcal{CN}(0,\sigma_{g,i}^2)$, $\alpha_{t,i} \sim \mathcal{CN}(0, \sigma_{t,i}^2)$, and $\alpha_{g,i}$ and $\alpha_{t,i}$ are independent with each other. Then, $|\alpha_{g,i}|^2 \sim \exp (\frac{1}{\sigma_{g,i}^2})$, $|\alpha_{t,i}|^2 \sim \exp (\frac{1}{\sigma_{t,i}^2})$, and $\mathbb{E}(|\alpha_{g,i}|^2 |\alpha_{t,i}|^2) = \sigma_{g,i}^2 \sigma_{t,i}^2$. Similarly, $\mathbb{E}(|\alpha_{h,i}|^2) = \sigma_{h,i}^2$.

\subsection{Proof of Lemma 2} \label{append_pro2}
When $\mathbf{Q}_k =  \frac{1}{KN_b} \mathbf{I}_{N_b}$, we have
\begin{subequations}
\begin{align}
R_k & = \mathbb{E}_{\mathbf{H}_{\text{eff},k}} \left[\log_2 \operatorname{det} \left( \mathbf{I}_{N_u}+ \frac{P_T}{\sigma^2} \mathbf{H}_{\text{eff},k} \mathbf{Q}_k \mathbf{H}_{\text{eff},k}^H \right)\right] \\
& = \mathbb{E}_{\mathbf{H}_{\text{eff},k}} \left[\log_2 \operatorname{det} \left( \mathbf{I}_{N_u}+ \frac{P_T}{\sigma^2 K N_b} \mathbf{H}_{\text{eff},k} \mathbf{H}_{\text{eff},k}^H \right)\right] \\
  &= \mathbb{E}_{\mathbf{H}_{\text{eff},k}}\left[\sum \limits_{i=1}^{N_s} \log _{2} \left( 1+ \frac{P_T}{\sigma^{2} K N_b} \lambda_i\left(\mathbf{H}_{\text{eff},k} \mathbf{H}_{\text{eff},k}^H\right) \right)\right].
\end{align}
\end{subequations}
The following derivations are similar to Appendix \ref{append_theorem1}. Thus, it is omitted here.

\end{appendices}

\bibliographystyle{IEEEtran}

\begin{thebibliography}{10}
\providecommand{\url}[1]{#1}
\csname url@samestyle\endcsname
\providecommand{\newblock}{\relax}
\providecommand{\bibinfo}[2]{#2}
\providecommand{\BIBentrySTDinterwordspacing}{\spaceskip=0pt\relax}
\providecommand{\BIBentryALTinterwordstretchfactor}{4}
\providecommand{\BIBentryALTinterwordspacing}{\spaceskip=\fontdimen2\font plus
\BIBentryALTinterwordstretchfactor\fontdimen3\font minus
  \fontdimen4\font\relax}
\providecommand{\BIBforeignlanguage}[2]{{%
\expandafter\ifx\csname l@#1\endcsname\relax
\typeout{** WARNING: IEEEtran.bst: No hyphenation pattern has been}%
\typeout{** loaded for the language `#1'. Using the pattern for}%
\typeout{** the default language instead.}%
\else
\language=\csname l@#1\endcsname
\fi
#2}}
\providecommand{\BIBdecl}{\relax}
\BIBdecl

\bibitem{6515173}
T.~S. Rappaport, S.~Sun, R.~Mayzus, H.~Zhao, Y.~Azar, K.~Wang, G.~N. Wong,
  J.~K. Schulz, M.~Samimi, and F.~Gutierrez, ``Millimeter wave mobile
  communications for {5G} cellular: It will work!'' \emph{IEEE Access}, vol.~1,
  pp. 335--349, 2013.

\bibitem{7400949}
R.~W. Heath, N.~Gonzalez-Prelcic, S.~Rangan, W.~Roh, and A.~M. Sayeed, ``An
  overview of signal processing techniques for millimeter wave {MIMO}
  systems,'' \emph{IEEE J. Sel. Topics Signal Process.}, vol.~10, no.~3, pp.
  436--453, 2016.

\bibitem{8910627}
Q.~Wu and R.~Zhang, ``Towards smart and reconfigurable environment: Intelligent
  reflecting surface aided wireless network,'' \emph{IEEE Commun. Mag.},
  vol.~58, no.~1, pp. 106--112, Jan. 2020.

\bibitem{9122596}
S.~Gong, X.~Lu, D.~T. Hoang, D.~Niyato, L.~Shu, D.~I. Kim, and Y.-C. Liang,
  ``Toward smart wireless communications via intelligent reflecting surfaces: A
  contemporary survey,'' \emph{IEEE Commun. Surv. Tutor.}, vol.~22, no.~4, pp.
  2283--2314, Jun. 2020.

\bibitem{9229054}
R.~Alghamdi, R.~Alhadrami, D.~Alhothali, H.~Almorad, A.~Faisal, S.~Helal,
  R.~Shalabi, R.~Asfour, N.~Hammad, A.~Shams, N.~Saeed, H.~Dahrouj, T.~Y.
  Al-Naffouri, and M.-S. Alouini, ``Intelligent surfaces for {6G} wireless
  networks: A survey of optimization and performance analysis techniques,''
  \emph{IEEE Access}, vol.~8, pp. 202\,795--202\,818, Oct. 2020.

\bibitem{8811733}
Q.~Wu and R.~Zhang, ``Intelligent reflecting surface enhanced wireless network
  via joint active and passive beamforming,'' \emph{IEEE Trans. Wireless
  Commun.}, vol.~18, no.~11, pp. 5394--5409, Nov. 2019.

\bibitem{8982186}
H.~Guo, Y.-C. Liang, J.~Chen, and E.~G. Larsson, ``Weighted sum-rate
  maximization for reconfigurable intelligent surface aided wireless
  networks,'' \emph{IEEE Trans. Wireless Commun.}, vol.~19, no.~5, pp.
  3064--3076, May 2020.

\bibitem{9110912}
S.~Zhang and R.~Zhang, ``Capacity characterization for intelligent reflecting
  surface aided {MIMO} communication,'' \emph{IEEE J. Sel. Areas Commun.},
  vol.~38, no.~8, pp. 1823--1838, Aug. 2020.

\bibitem{9234098}
P.~Wang, J.~Fang, L.~Dai, and H.~Li, ``Joint transceiver and large intelligent
  surface design for massive {MIMO} mmwave systems,'' \emph{IEEE Trans.
  Wireless Commun.}, vol.~20, no.~2, pp. 1052--1064, Feb. 2021.

\bibitem{huawei}
\BIBentryALTinterwordspacing
\emph{{6G} The Next Horizon}.\hskip 1em plus 0.5em minus 0.4em\relax Huawei 6G
  White Paper, 2021. [Online]. Available:
  \url{https://www-file.huawei.com/-/media/corp2020/pdf/tech-insights/1/6g-white-paper-cn.pdf?la=zh.}
\BIBentrySTDinterwordspacing

\bibitem{8354789}
B.~Wang, F.~Gao, S.~Jin, H.~Lin, and G.~Y. Li, ``Spatial- and
  frequency-wideband effects in millimeter-wave massive {MIMO} systems,''
  \emph{IEEE Trans. Signal Process.}, vol.~66, no.~13, pp. 3393--3406, Jul.
  2018.

\bibitem{ning2021prospective}
\BIBentryALTinterwordspacing
B.~Ning, Z.~Tian, Z.~Chen, C.~Han, J.~Yuan, and S.~Li, ``Prospective
  beamforming technologies for ultra-massive {MIMO} in terahertz
  communications: A tutorial,'' 2021. [Online]. Available:
  \url{https://arxiv.org/abs/2107.03032}
\BIBentrySTDinterwordspacing

\bibitem{9409636}
S.~Ma, W.~Shen, J.~An, and L.~Hanzo, ``Wideband channel estimation for
  {IRS}-aided systems in the face of beam squint,'' \emph{IEEE Trans. Wireless
  Commun.}, vol.~20, no.~10, pp. 6240--6253, 2021.

\bibitem{9417413}
Y.~Chen, D.~Chen, and T.~Jiang, ``Beam-squint mitigating in reconfigurable
  intelligent surface aided wideband mmwave communications,'' in \emph{Proc.
  IEEE WCNC}, Apr. 2021, pp. 1--6.

\bibitem{9039554}
Y.~Yang, B.~Zheng, S.~Zhang, and R.~Zhang, ``Intelligent reflecting surface
  meets {OFDM}: Protocol design and rate maximization,'' \emph{IEEE Trans.
  Commun.}, vol.~68, no.~7, pp. 4522--4535, Jul. 2020.

\bibitem{8964457}
Y.~Yang, S.~Zhang, and R.~Zhang, ``{IRS}-enhanced {OFDMA}: Joint resource
  allocation and passive beamforming optimization,'' \emph{IEEE Wireless
  Commun. Lett.}, vol.~9, no.~6, pp. 760--764, Jun. 2020.

\bibitem{8937491}
B.~Zheng and R.~Zhang, ``Intelligent reflecting surface-enhanced {OFDM}:
  Channel estimation and reflection optimization,'' \emph{IEEE Wireless Commun.
  Lett.}, vol.~9, no.~4, pp. 518--522, 2020.

\bibitem{9610122}
Z.~He, H.~Shen, W.~Xu, and C.~Zhao, ``Low-cost passive beamforming for
  {RIS}-aided wideband {OFDM} systems,'' \emph{IEEE Wireless Commun. Lett.},
  vol.~11, no.~2, pp. 318--322, 2022.

\bibitem{9120639}
H.~Li, R.~Liu, M.~Liy, Q.~Liu, and X.~Li, ``{IRS}-enhanced wideband
  {MU-MISO-OFDM} communication systems,'' in \emph{Proc. IEEE WCNC}, May 2020,
  pp. 1--6.

\bibitem{nuti2021spectral}
P.~Nuti, E.~Balti, and B.~L. Evans, ``Spectral efficiency optimization for
  mmwave wideband {MIMO} {RIS}-assisted communication,'' in \emph{2022 IEEE
  95th Vehicular Technology Conference: (VTC2022-Spring)}, 2022, pp. 1--6.

\bibitem{8964330}
K.~Ying, Z.~Gao, S.~Lyu, Y.~Wu, H.~Wang, and M.-S. Alouini, ``{GMD}-based
  hybrid beamforming for large reconfigurable intelligent surface assisted
  millimeter-wave massive {MIMO},'' \emph{IEEE Access}, vol.~8, pp.
  19\,530--19\,539, Jan. 2020.

\bibitem{hong2022hybrid}
S.~H. Hong, J.~Park, S.-J. Kim, and J.~Choi, ``Hybrid beamforming for
  intelligent reflecting surface aided millimeter wave {MIMO} systems,''
  \emph{IEEE Trans. Wireless Commun.}, pp. 1--1, 2022.

\bibitem{9459505}
Z.~Zhang and L.~Dai, ``A joint precoding framework for wideband reconfigurable
  intelligent surface-aided cell-free network,'' \emph{IEEE Trans. Signal
  Process.}, vol.~69, pp. 4085--4101, Jun. 2021.

\bibitem{9685734}
W.~Jiang, B.~Chen, S.~Garg, J.~Nie, J.~Zhao, and Z.~Xiong, ``Joint transmit
  precoding and reflect beamforming for {IRS}-assisted {MIMO-OFDM} secure
  communications,'' in \emph{Proc. IEEE GLOBECOM}, Dec. 2021, pp. 1--6.

\bibitem{9520295}
W.~Jiang, B.~Chen, J.~Zhao, Z.~Xiong, and Z.~Ding, ``Joint active and passive
  beamforming design for the {IRS}-assisted {MIMOME-OFDM} secure
  communications,'' \emph{IEEE Trans. Veh. Technol.}, vol.~70, no.~10, pp.
  10\,369--10\,381, Oct. 2021.

\bibitem{li2022risassisted}
R.~Li, S.~Sun, Y.~Chen, C.~Han, and M.~Tao, ``Ergodic achievable rate analysis
  and optimization of {RIS}-assisted millimeter-wave {MIMO} communication
  systems,'' \emph{IEEE Trans. Wireless Commun.}, pp. 1--1, 2022.

\bibitem{9389801}
H.~Li, W.~Cai, Y.~Liu, M.~Li, Q.~Liu, and Q.~Wu, ``Intelligent reflecting
  surface enhanced wideband {MIMO-OFDM} communications: From practical model to
  reflection optimization,'' \emph{IEEE Trans. Commun.}, vol.~69, no.~7, pp.
  4807--4820, 2021.

\bibitem{7397861}
X.~Yu, J.-C. Shen, J.~Zhang, and K.~B. Letaief, ``Alternating minimization
  algorithms for hybrid precoding in millimeter wave {MIMO} systems,''
  \emph{IEEE J. Sel. Topics Signal Process.}, vol.~10, no.~3, pp. 485--500,
  Apr. 2016.

\bibitem{8844787}
Y.~Chen, D.~Chen, T.~Jiang, and L.~Hanzo, ``Channel-covariance and
  angle-of-departure aided hybrid precoding for wideband multiuser millimeter
  wave {MIMO} systems,'' \emph{IEEE Trans. Commun.}, vol.~67, no.~12, pp.
  8315--8328, Dec. 2019.

\bibitem{8794743}
M.~Wang, F.~Gao, S.~Jin, and H.~Lin, ``An overview of enhanced massive {MIMO}
  with array signal processing techniques,'' \emph{IEEE J. Sel. Topics Signal
  Process.}, vol.~13, no.~5, pp. 886--901, Sep. 2019.

\bibitem{6847111}
A.~Alkhateeb, O.~El~Ayach, G.~Leus, and R.~W. Heath, ``Channel estimation and
  hybrid precoding for millimeter wave cellular systems,'' \emph{IEEE J. Sel.
  Topics Signal Process.}, vol.~8, no.~5, pp. 831--846, Oct. 2014.

\bibitem{marshall1979inequalities}
A.~W. Marshall, I.~Olkin, and B.~C. Arnold, \emph{Inequalities: theory of
  majorization and its applications}.\hskip 1em plus 0.5em minus 0.4em\relax
  Springer, 1979, vol. 143.

\bibitem{1611096}
A.~Tulino, A.~Lozano, and S.~Verdu, ``Capacity-achieving input covariance for
  single-user multi-antenna channels,'' \emph{IEEE Trans. Wireless Commun.},
  vol.~5, no.~3, pp. 662--671, Mar. 2006.

\bibitem{7501500}
M.~K. Samimi and T.~S. Rappaport, ``3-{D} millimeter-wave statistical channel
  model for {5G} wireless system design,'' \emph{IEEE Trans. Microw. Theory
  Techn.}, vol.~64, no.~7, pp. 2207--2225, Jul. 2016.

\bibitem{1045248}
P.~Almers, F.~Tufvesson, O.~Edfors, and A.~Molisch, ``Measured capacity gain
  using water filling in frequency selective {MIMO} channels,'' in \emph{IEEE
  Int. Symp. Personal, Indoor and Mobile Radio Commun.}, vol.~3, 2002, pp.
  1347--1351 vol.3.

\bibitem{AbsilMahonySepulchre+2009}
\BIBentryALTinterwordspacing
P.-A. Absil, R.~Mahony, and R.~Sepulchre, \emph{Optimization Algorithms on
  Matrix Manifolds}.\hskip 1em plus 0.5em minus 0.4em\relax Princeton
  University Press, 2009. [Online]. Available:
  \url{https://doi.org/10.1515/9781400830244}
\BIBentrySTDinterwordspacing

\bibitem{6831723}
J.~Chen, ``When does asymptotic orthogonality exist for very large arrays?'' in
  \emph{Proc. IEEE GLOBECOM}, Dec. 2013, pp. 4146--4150.

\bibitem{8816689}
X.~Yang, X.~Li, S.~Zhang, and S.~Jin, ``On the ergodic capacity of mmwave
  systems under finite-dimensional channels,'' \emph{IEEE Trans. Wireless
  Commun.}, vol.~18, no.~11, pp. 5440--5453, Nov. 2019.

\end{thebibliography}

\end{document}